%% file: main.tex
\documentclass[nohyperref]{article}

\usepackage{microtype}
\usepackage{graphicx}
\usepackage{subcaption}
\usepackage{booktabs} 

\usepackage{hyperref}

\usepackage[accepted]{icml2023}

\usepackage{amsmath}
\usepackage{amssymb}
\usepackage{mathtools}
\usepackage{amsthm}

\usepackage[capitalize,noabbrev]{cleveref}

\theoremstyle{plain}
\newtheorem{theorem}{Theorem}

\newtheorem{lemma}[theorem]{Lemma}

\theoremstyle{definition}

\theoremstyle{remark}
\newtheorem{remark}{Remark}

\usepackage[textsize=tiny]{todonotes}
\usepackage{dsfont} 

\usepackage{float} 

\begin{document}

\twocolumn[
\icmltitle{
Robust Non-Linear Feedback Coding via Power-Constrained Deep Learning} 
\thispagestyle{plain}
\pagestyle{plain}





\begin{icmlauthorlist}
\icmlauthor{Junghoon Kim}{Purdue}
\icmlauthor{Taejoon Kim}{Kansas}
\icmlauthor{David Love}{Purdue}
\icmlauthor{Christopher Brinton}{Purdue}
\end{icmlauthorlist}

\icmlaffiliation{Purdue}{Electrical and Computer Engineering, Purdue University, West Lafayette, IN, USA}
\icmlaffiliation{Kansas}{Electrical Engineering and Computer Science, University of Kansas, Lawrence, KS, USA}

\icmlcorrespondingauthor{Junghoon Kim}{kim3220@purdue.edu}

\icmlkeywords{Machine Learning, ICML}

\vskip 0.3in
]



\printAffiliationsAndNotice{}  

\begin{abstract}
\input{abstract}
\end{abstract}

\input{intro}

\input{system}


\input{alg}

\input{sim}

\input{conc}

\input{ack}

\clearpage

\bibliography{example_paper}
\bibliographystyle{icml2023}

\newpage

\input{appendix}


\end{document}

%% file: abstract.tex

The design of codes for feedback-enabled communications has been a long-standing open problem.
Recent research on non-linear, deep learning-based coding schemes have demonstrated significant improvements in communication reliability over linear codes, but are still vulnerable to the presence of forward and feedback noise over the channel. In this paper, we develop a new family of non-linear feedback codes that greatly enhance robustness to channel noise. Our autoencoder-based architecture is designed to learn codes based on consecutive blocks of bits, which obtains de-noising advantages over bit-by-bit processing to help overcome the physical separation between the encoder and decoder over a noisy channel. Moreover, we develop a power control layer at the encoder to explicitly incorporate hardware constraints into the learning optimization, and prove that the resulting average power constraint is satisfied asymptotically. Numerical experiments demonstrate that our scheme outperforms state-of-the-art feedback codes by wide margins over practical forward and feedback noise regimes, and provide information-theoretic insights on the behavior of our non-linear codes. Moreover, we observe that, in a long blocklength regime, canonical error correction codes are still preferable to feedback codes when the feedback noise becomes high. Our code is available at  
\href{https://anonymous.4open.science/r/RCode1}{https://anonymous.4open.science/r/RCode1}.

%% file: intro.tex
\section{Introduction}
\vspace{-1mm}


The design of codes has been an important field of research both in information theory and communications,
targeting
efficient and reliable data transmission across a noisy channel.
We now have near-optimal codes for 
the canonical open-loop
additive white Gaussian noise (AWGN) channel setting first proposed by Shannon~\cite{shannon1948mathematical},
thanks to the extraordinary advancements in code design over a number of decades.
%
%
However, the design of practical codes 
for a variety of other important channel models
has remained long-standing open problems.
%
In particular, when
\textit{feedback} is available in communication systems, i.e., when a transmitter can obtain information on the receive signals through a reverse link from a receiver,
it has been shown that while capacity cannot be increased~\cite{shannon1956zero}, communication reliability can be improved by utilizing  \textit{feedback codes}~\cite{schalkwijk1966coding, butman1969general,shayevitz2011optimal, ben2017interactive}. 
However, the design of feedback codes 
is non-trivial 
since both the bit stream and feedback information (i.e., past receive signals) should be incorporated into the design.  Though feedback codes hold the potential to revolutionize future communication systems, major innovations are needed in their design and implementation.







\vspace{-1mm}
\subsection{Linear Feedback Coding}
\vspace{-1mm}

Over several decades, research on the design of feedback codes for closed-loop AWGN channels mostly focused on the \textit{linear} family of codes, which simplifies code design.
The seminal work~\cite{schalkwijk1966coding}  introduced a linear coding technique for AWGN channels with noiseless feedback, known as the SK scheme, which achieves doubly exponential decay in the probability of error.
However, for noisy feedback, the SK scheme does not perform well.
In response, \cite{chance2011concatenated} introduced a linear coding scheme for AWGN channels with noisy feedback, known as the CL scheme. The CL scheme was further examined and found to be optimal within the linear family of codes under the noisy feedback scenario~\cite{agrawal2011iteratively}.
There have been attempts to view the linear feedback code design as 
feedback stabilization in control theory~\cite{elia2004bode} and
dynamic programming~\cite{mishra2021linear}.
However, the linear assumption made in these works 
severely limits their ability to produce optimal codes.


\subsection{Deep Learning-Based Channel Coding}

A recent trend of research has been examining code design from a deep learning perspective to take advantage of its non-linear structure.
%
%
%
Neural encoders and decoders
have been shown to improve communication reliability and/or efficiency for various canonical channel settings including open-loop AWGN channels~\cite{nachmani2016learning, o2017introduction, dorner2017deep, kim2018communication, jiang2019turbo, habib2020learning, makkuva2021ko, chen2021cyclically}.
%
In the closed-loop AWGN channel case,
Deepcode~\cite{kim2018deepcode}  proposes an autoencoder architecture to generate non-linear feedback codes. Deepcode  was shown to outperform SK and CL in terms of error performance across many noise scenarios due to the wider degree of flexibility that non-linearity provides for the creation of feedback codes.
Deep extended feedback (DEF) codes~\cite{safavi2021deep} generalize Deepcode by including parity symbols generated based on forward-channel output observations over longer time intervals and supporting high-order modulation in the encoder to maximize spectral efficiency.
Generalized block attention feedback (GBAF) codes~\cite{ozfatura2022all} have recently been proposed with self-attention modules that can incorporate different neural network architectures. It has been shown that GBAF codes significantly outperform the existing solutions, especially in the noiseless feedback scenario.


Nevertheless, the vulnerability of these feedback codes to high forward/feedback noise remains understudied. High noise settings have become more pervasive as wireless networks have become denser, making reliable communications even more dependent on channel feedback~\cite{shafi20175g}. As pointed out in existing works~\cite{o2017introduction, dorner2017deep, jiang2019turbo}, end-to-end learning for the design of codes over point-to-point channels benefits significantly from an autoencoder architecture’s ability to jointly train the encoder and decoder. In this respect, high noise regimes pose two central challenges:
 
\textbf{(1) Encoder-decoder mismatch.} The encoder (as a transmitter) and decoder (as a receiver) are implemented on two separate platforms. 
Channel noise therefore causes mismatches in the latent space for coding between encoder and decoder which cannot be directly calibrated due to the limited bandwidth of the forward and feedback links.
In Deepcode~\cite{kim2018deepcode}, the encoding structure consists of two distinct phases and operates bit-by-bit, which limits the size of the latent space available to build resilience against high noise conditions. 
In this work, we consider finite-length bit streams as the units for autoencoder learning to maximally benefit from \textit{noise averaging} that forms the basis for error correction codes~\cite{clark2013error}, and show the robustness of our codes to high noise levels.




\textbf{(2) Inefficient power allocation.} The transmitter has intrinsic hardware limitations which constrain the encoding outputs in terms of power across channel uses. Coding schemes for point-to-point channels without feedback~\cite{o2017introduction, dorner2017deep, jiang2019turbo} exploit normalization at the encoding outputs to satisfy the power constraint. For power allocation in feedback-enabled channels, Deepcode~\cite{kim2018deepcode} employs two layers of power weights in addition to a normalization layer. 
However, none of these approaches account for the impact of channel noise on the efficacy of power allocation. 
In this work, we show how power control can be explicitly incorporated into the encoder optimization procedure to obtain constraint satisfaction guarantees.

\subsection{Summary of Contributions}

%
\begin{itemize}
    \item 
    We develop a recurrent neural network (RNN) autoencoder-based architecture for power-constrained, feedback-enabled communications. Using this architecture, we suggest a new class of non-linear feedback codes that build robustness to forward and feedback noise in AWGN channels with feedback.
    \item 
    Our learning architecture addresses the challenge of encoder-decoder separation over noisy channels by considering the entire bit stream as a single unit to potentially benefit from noise averaging, analogously to  error correction codes.
    We also adopt a bi-directional attention-based decoding architecture  to fully exploit  correlations among noisy receive signals.
    \item 
    We augment our encoder architecture with a power control layer, which we prove satisfies power constraints asymptotically. 
    We also provide information-theoretic insights on the power distribution obtained from our non-linear feedback codes, 
    showing that power allocation is highest for early channel uses and then tapers off over time, similar to optimal linear codes.
    \item Through numerical experiments, we show that our methodology can outperform state-of-the-art feedback coding schemes by wide margins.
    While other feedback codes are vulnerable to high feedback noise, our codes still perform well as long as the forward noise is reasonable.
    Also, unlike existing codes, our methodology draws significant benefit from reductions in feedback noise even when the forward noise is high. 
    \item 
    We propose a modulo-based approach to extend our finite block length coding architecture to long block lengths. This provides computational scalability and generalization across different block lengths without requiring any re-training or parameter-tuning.
\end{itemize}







%% file: system.tex
\section{System Model and Optimization}



\begin{figure*}[ht]
  \includegraphics[width=.675\linewidth]{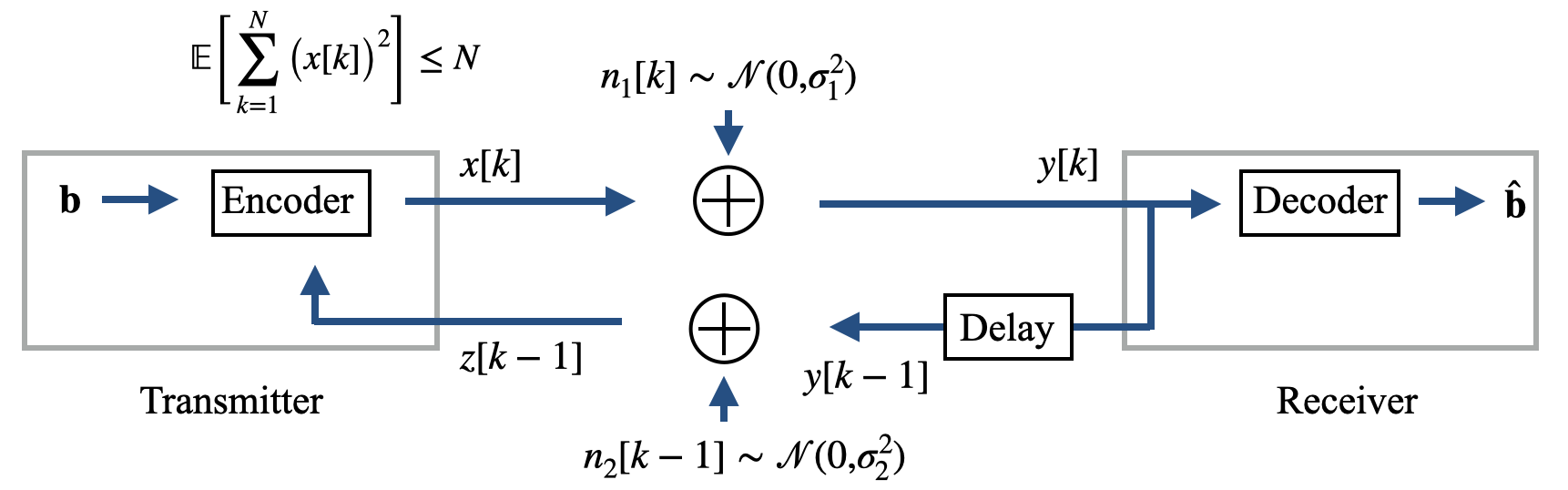}
  \centering
  \vspace{-2mm}
  \caption{A point-to-point AWGN communication channel with noisy feedback.}
  \label{fig:system_merged}
  \vspace{-2mm}
\end{figure*} 

\subsection{Transmission Model}
\label{ssec:trans}
We consider a canonical point-to-point AWGN communication channel with noisy feedback as shown in Figure~\ref{fig:system_merged}.
We assume that 
the transmission occurs over $N$ channel uses (timesteps).
Let $k \in \{1,...,N\}$ denote the index of channel use and $x[k] \in \mathbb{R}$ represent the transmit signal at time $k$.
At time $k$, the receiver receives 
the signal
\begin{equation}
    y[k] = x[k] + n_1[k] \in \mathbb{R}, \quad k=1,...,N,
    \label{eq:y}
\end{equation}
where $n_1[k] \sim \mathcal{N}(0,\sigma_1^2)$ is Gaussian noise for the forward channel.
We consider an average power constraint as
\begin{equation}
    \mathbb{E}\bigg[ \sum_{k=1}^N  \big(x[k] \big)^2   \bigg] \le N.
\end{equation}

At each time $k$, the receiver feeds back its receive signal $y[k]$ to the transmitter over a noisy feedback channel, as shown in Figure 1. The transmitter then receives
\begin{align}
    z[k] = y[k] + n_2[k],
    \quad k=1,...,N,
    \label{eq:z}
\end{align}
where $n_2[k] \sim \mathcal{N}(0,\sigma_2^2)$ is the feedback noise.

\subsection{Functional Form of Encoding and Decoding}
\label{ssec:enc/dec}
The goal of the transmission is to successfully deliver a bit stream ${\bf b} \in \{0,1\}^K$ from the transmitter to the receiver over a noisy channel, where $K$ is the number of source bits.
For efficient communication of ${\bf b}$, we consider the following transmitter encoding and receiver decoding procedures.

\textbf{Encoding.}
The transmitter encodes the bit stream ${\bf b} \in \{0,1\}^K$ to generate the transmit signals of $N$ channel uses, i.e., $\{{x}[k]\}_{k=1}^N$.
The coding rate is defined by $r=K/N$.
%
Provided feedback from the receiver, the encoding at the transmitter is described as a function of the bit stream ${\bf b}$ and the feedback signals $\{z[j]\}_{j=1}^{k-1}$ in \eqref{eq:z}. 
Defining an encoding function at time $k$ as $f_k: \mathbb{R}^{K+k-1} \to \mathbb{R}$, 
we represent the encoding at time $k$ as  
\begin{align}
    x[k] = f_k({\bf b}, z[1], ..., z[k-1]), \quad k=1,...,N.
    \label{eq:encoding}
\end{align}

\textbf{Decoding.} Once the $N$ transmissions are completed, the receiver decodes the receive signals of $N$ channel uses, i.e., $\{{y}[k]\}_{k=1}^N$ in \eqref{eq:y}, to obtain an estimate of the bit stream, $\hat {\bf b} \in \{0,1\}^K$.
%
Define a decoding function as $g:\mathbb{R}^N \to \{0,1\}^K$. Then, we represent the decoding process conducted by the receiver as
\begin{align}
    \hat {\bf b} = g(y[1], ..., y[N]).
\end{align}

\subsection{Optimization for Encoder and Decoder}
\label{ssec:opt}

Block error rate (BLER), the ratio of the number of incorrect bit streams to the total number of bit streams, is commonly used as a performance metric to assess the communication reliability of
bit stream transmissions.
Formally, with BLER defined as $\text{Pr}[ {\bf b} \neq \hat {\bf b} ]$,
we consider our encoder-decoder design objective to be the following optimization problem:
\begin{align}
&  \underset{ f_1,..,f_N, g } {\text{minimize}} & & 
 \text{Pr}[ {\bf b} \neq \hat {\bf b} ]
 \label{opt:obj}
\\
& \text{subject to}
& &  
 \mathbb{E}_{ {\bf b}, {\bf n}_1, {\bf n}_2 } \bigg[ \sum_{k=1}^N \big( x[k] \big)^2  \bigg] \le N
 \label{opt:const:power}
\end{align}

The expectation in \eqref{opt:const:power} is  taken over the distributions of the bit stream $ {\bf b}$ and the noises, ${\bf n}_1$ and ${\bf n}_2$, since the encoding output $x[k]$ depends on $ {\bf b}$, ${\bf n}_1$, and ${\bf n}_2$ in \eqref{eq:encoding}, where
 ${\bf n}_1 = [n_1[1], ..., n_1[N]]^T$ and ${\bf n}_2 = [n_2[1], ..., n_2[N]]^T$.
 %

Thus, the goal is to optimize $N$ encoding functions $\{f_k\}_{k=1}^N$ and a decoding function $g$.
However, the complexity of designing $N$ encoding functions increases to a great extent with the number of channel uses $N$.
To mitigate the associated design complexity,
we introduce a state propagation-based encoding technique, discussed next.





\subsection{State Propagation-Based Encoding}
\label{ssec:state_enc}

Given that the inputs used for encoding at each time in \eqref{eq:encoding} are overlapping, 
we expect that the $N$ encoding functions are correlated with one another. 
As a result,
we consider a state propagation-based encoding technique where the encoding is performed at each timestep using only two distinct functions: (i) signal-generation and (ii) state-propagation.


Defining the signal-generation function as $f: \mathbb{R}^{K+N_s+1} \to \mathbb{R} $,
we re-write the encoding process in \eqref{eq:encoding} as
\vspace{-.5mm}
\begin{align}
    x[k] = f({\bf b}, z[k-1], {\bf s}[k]),
    \quad k=1,...,N,
    \label{eq:func_f}
\end{align}
where we assume $z[-1]=0$. Here, ${\bf s}[k] \in \mathbb{R}^{N_s} $ is the state vector, which propagates over time through the state-propagation function $h:\mathbb{R}^{K+N_s+1} \to \mathbb{R}^{N_s}$, given by
\begin{align}
    {\bf s}[k] = h({\bf b}, z[k-1], {\bf s}[k-1]), \quad k=1,...,N.
    \label{eq:func_h}
\end{align}
Note that, in \eqref{eq:func_h}, the current state ${\bf s}[k]$ is updated from the prior state ${\bf s}[k-1]$ by incorporating ${\bf b}$ and $z[k-1]$. For the initial condition, we assume ${\bf s}[-1]={\bf 0}$.
This encoding model in \eqref{eq:func_f}-\eqref{eq:func_h} can be seen as a general and non-linear extension of the state-space model used for linear encoding in feedback systems~\cite{elia2004bode}.


Through this technique, we only need to design the two encoding functions, $f$ and $h$ -- instead of $N$ separate functions -- and the decoding function $g$. We thus re-write the optimization problem in \eqref{opt:obj}-\eqref{opt:const:power} as
\begin{align}
&  \underset{ f, h, g } {\text{minimize}} & & 
 \text{Pr}[ {\bf b} \neq \hat {\bf b} ]
\label{opt:obj:state-based}
\\
\vspace{-.5mm}
& \text{subject to}
& &  
\mathbb{E}_{ {\bf b}, {\bf n}_1, {\bf n}_2 } \bigg[ \sum_{k=1}^N \big( x[k] \big)^2  \bigg] \le N.
\label{opt:const:power:state-based}
\end{align}
\vspace{-4mm}

Nevertheless, the problem 
\eqref{opt:obj:state-based}-\eqref{opt:const:power:state-based} is non-trivial since the functions, $f$, $h$, and  $g$, can take arbitrary forms. 
Over several decades of research in the design of such functions,
a linear assumption was made~\cite{schalkwijk1966coding, butman1969general, chance2011concatenated}.\footnote{The existing works, such as SK, CL, and Deepcode, can be interpreted in the framework of state propagation-based encoding.}
While omitting low-complexity schemes, this constrains the degree of freedom in the design of such functions, leading to unsatisfactory error performances~\cite{kim2018deepcode}.
We next employ this state propagation-based encoding in the design of non-linear feedback codes that are robust to channel noise.




%% file: alg.tex
\section{Feedback Coding Methodology}
Our overall coding scheme is depicted in Figure \ref{fig:overall}. It follows the RNN autoencoder-based architecture detailed below.
\vspace{-.5mm}


\begin{figure*}[ht]
  \centering
  \vspace{-2mm}
  \includegraphics[width=.9\linewidth]{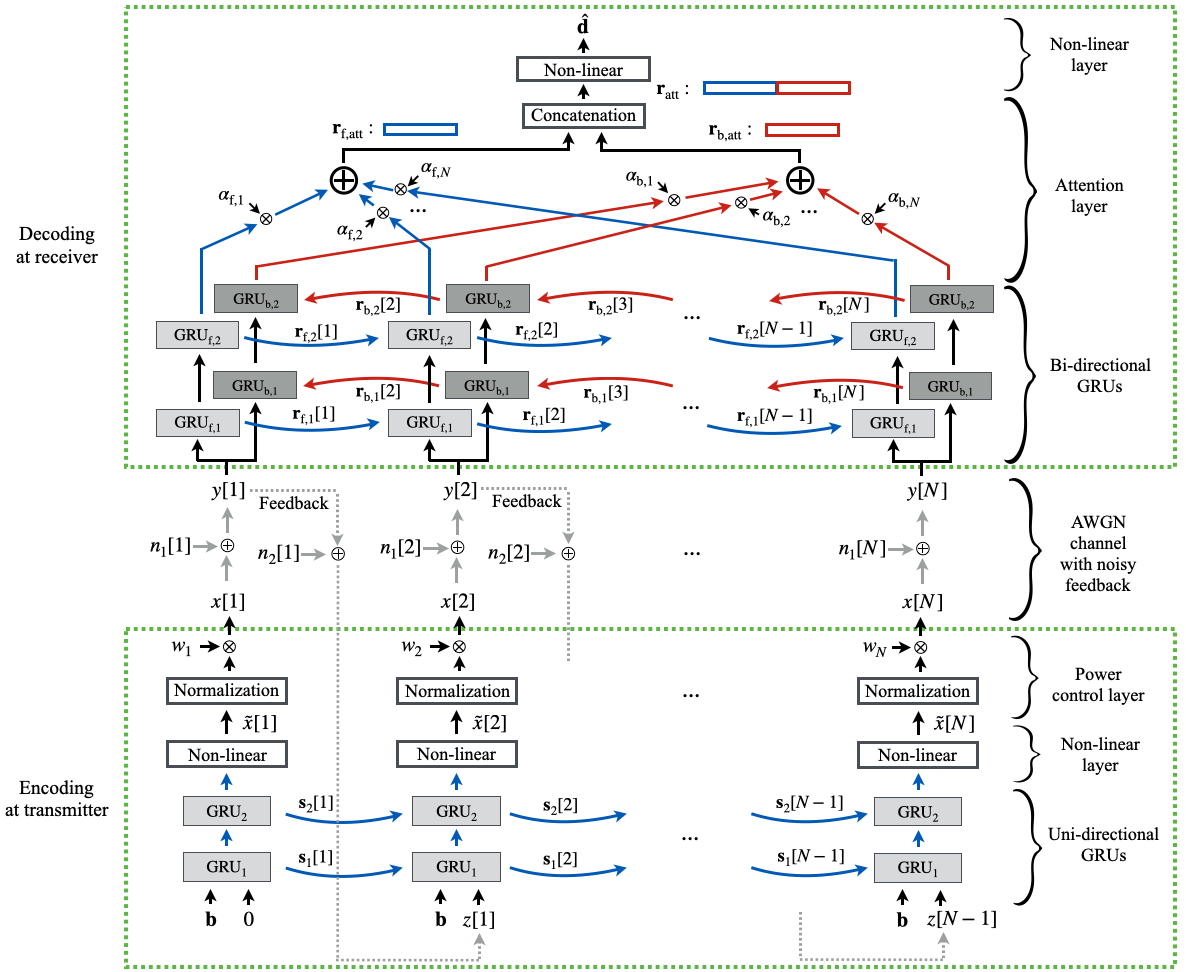}
  \vspace{-3mm}
  \caption{Our proposed RNN autoencoder-based architecture for non-linear coding over an AWGN channel with noisy feedback. At the end of the decoding process, $\hat {\bf d}$ denotes the probability distribution of $2^K$ possible outcomes of the bit stream estimate $\hat {\bf b}$.
  A compact form of the architecture is included in Appendix~\ref{sec:app:compact}.}
  \label{fig:overall}
  \vspace{-2mm}
\end{figure*}

\subsection{Encoding}
\label{ssec:encoding}

We follow the state propagation-based encoding  approach discussed in Section~\ref{ssec:state_enc}.
The state-propagation function $h$ in \eqref{eq:func_h} consists of two layers of gated
recurrent units (GRUs), while the signal-generation function $f$ in \eqref{eq:func_f} consists of a non-linear layer and a power control layer in  sequence. 


\textbf{GRUs for state propagation.}
We adopt two layers of  unidirectional GRUs to capture the time correlation of the feedback signals in a causal manner. 
Formally, we represent the input-output relationship at each layer at time $k$ as
\begin{align}
    {\bf s}_1[k] & = \text{GRU}_1({\bf b}, z[k-1], {\bf s}_1[k-1]),
    \label{eq:s1}
    \\
    {\bf s}_2[k] & = \text{GRU}_2({\bf s}_1[k], {\bf s}_2[k-1]),
    \label{eq:s2}
\end{align}
where $\text{GRU}_i$  represents a functional form of GRU processing at layer $i$ and ${\bf s}_i[k] \in \mathbb{R}^{N_{\text{s},i}}$ is the state vector obtained by $\text{GRU}_i$ at time $k$, where $i=1,2$ and $k=1,...,N$. 
For the initial conditions, we assume ${\bf s}_i[-1]={\bf 0}$, $i=1,2$.

Equations \eqref{eq:s1}-\eqref{eq:s2} can be represented as a functional form of the state propagation-based encoding in \eqref{eq:func_h}. By defining the overall state vector as ${\bf s}[k] = [{\bf s}_1[k], {\bf s}_2[k]]$, we obtain ${\bf s}[k] = h({\bf b}, z[k-1], {\bf s}[k-1])$, where $h$ implies the process of two layers of GRUs in \eqref{eq:s1}-\eqref{eq:s2}.
Note that ${\bf s}[k]$ propagates over time through the GRUs by incorporating the current input information into its state.
Because the bit stream ${\bf b}$ with length $K$ is handled as a \textit{block} to generate transmit signals with any length $N$, 
our method is flexible enough to support any coding rate $r$, unlike the prior approach~\cite{kim2018deepcode} that only appears to support~$r=1/3$.

\textbf{Non-linear layer.}
We adopt an additional non-linear layer at the output of the GRUs. The state vector at the last layer of the GRUs, i.e., ${\bf s}_2[k]$, is taken as an input to this additional non-linear layer.
Formally, we can represent the process of the non-linear layer as
\begin{align}
    \tilde x[k] = \phi( {\bf w}_e^T {\bf s}_2[k] + {b}_e  ), \quad k=1,...,N,
    \label{eq:enc:non-linear}
\end{align}
where ${\bf w}_e \in \mathbb{R}^{N_{\text{s},2}}$ and ${b}_e \in \mathbb{R}$ are the trainable weights and biases, respectively, and $\phi:
\mathbb{R} \to \mathbb{R}$ is an activation function (hyperbolic tangent).

It is possible to
use $\tilde x[k]$ directly as a transmit signal, since $\tilde x[k]$ ranges in $(-1,1)$
and satisfies the power constraint $\sum_{k=1}^N  ( \tilde x[k] )^2 \le N$.
However, this does not ensure maximum utilization of the transmit power budget. Power control over the sequence of transmit signals is essential in the design of encoding schemes for feedback-enabled communications in order to achieve robust error performance~\cite{schalkwijk1966coding, chance2011concatenated,kim2018deepcode}.

\textbf{Power control layer.} We introduce a power control layer to optimize for the power distribution, while also satisfying the power constraint in \eqref{opt:const:power:state-based}. This layer consists of two consecutive modules: (i) normalization and (ii) power-weight multiplication.
The transmit signal at time $k$
is then generated by
\begin{align}
    {x}[k] = w_k  \gamma^{(J)}_k(\tilde x[k]), \quad k=1,...,N,
    \label{eq:norm_power}
\end{align}
where 
$\gamma^{(J)}_k: \mathbb{R} \rightarrow \mathbb{R}$ is a normalization function applied to $\tilde x[k]$, which consists of the sample mean and sample variance calculated from the data with size $J$.
Here, $w_k$ is a trainable power weight satisfying $\sum_{k=1}^N w_k^2 = N$.

Through the power control layer, 
the power weights are optimized via training in a way that minimizes the BLER in \eqref{opt:obj:state-based}.
At the same time, the power constraint in \eqref{opt:const:power:state-based} should be satisfied.
To obtain a smaller BLER, it is advantageous  
to ensure maximum utilization of the power budget $N$~\cite{schalkwijk1966coding, chance2011concatenated}.
However, satisfying the power constraint in 
an (ensemble) average sense 
is non-trivial, since the distributions of $\{x[k]\}_{k=1}^N$ are unknown.
Therefore, we approach it in an empirical sense:
(i) During training, we use standard batch normalization; we normalize ${\tilde x}[k]$ with the sample mean and sample variance calculated from each batch of data (with size $N_\text{batch}$) at each $k$. (ii) After training, we calculate and save the sample mean and sample variance from the entire training data (with size $J$).
(iii) For inference, we use the saved mean and variance for normalization.

In the following lemma, we show that the above procedure guarantees satisfaction of the equality power constraint 
in an asymptotic sense
with a large number of training data used for normalization.

\begin{lemma}
    Given the power control layer in \eqref{eq:norm_power}, the power constraint in \eqref{opt:const:power:state-based} 
    converges to $N$
    almost surely,
    i.e., $\mathbb{E}_{{\bf b}, {\bf n}_1, {\bf n}_2 } \big[ \sum_{k=1}^N ( {x}[k] )^2 \big]  \xrightarrow{a.s.} N$,
    as the number of training data $J$ used for normalization tends to infinity.
\end{lemma} 
\begin{proof}
See Appendix~\ref{sec:app:proof_Lemma1}.
\end{proof}

\begin{remark}
    Because of the data-dependent normalization in generating ${x}[k]$ in \eqref{eq:norm_power},
$\mathbb{E}_{{\bf b}, {\bf n}_1, {\bf n}_2 } \big[ \sum_{k=1}^N (  {x}[k] )^2 \big] $ is a random sequence along $J$ in our implementation.
We note that our neural network builds resiliency to the small size of inference data  by using the saved mean and variance (rather than calculating them with the batch of inference data).
\end{remark}

\subsection{Decoding}
\label{ssec:decoding}

As shown in Figure \ref{fig:overall}, our decoding function $g$ consists of the bi-directional GRUs, the attention layer, and the non-linear layer. We discuss each of them in detail below.



\textbf{Bi-directional GRU.}
We utilize two layers of bi-directional GRUs to capture the time correlation of the receive signals both in the forward and backward directions. 
We represent the input-output relationship of the forward directional GRUs at time $k$ as
\begin{align}
    {\bf r}_{\text{f},1}[k] & = \text{GRU}_{\text{f},1} (y[k],{\bf r}_{\text{f},1}[k-1] ),
    \nonumber
    \\
    {\bf r}_{\text{f},2}[k] & = \text{GRU}_{\text{f},2} ({\bf r}_{\text{f},1}[k],{\bf r}_{\text{f},2}[k-1] ),
    \label{eq:dec:rf1}
\end{align}
and that of the backward directional GRUs as
\begin{align}
    {\bf r}_{\text{b},1}[k] & = \text{GRU}_{\text{b},1} (y[k],{\bf r}_{\text{b},1}[k+1] ),
    \nonumber
    \\
    {\bf r}_{\text{b},2}[k] & = \text{GRU}_{\text{b},2} ({\bf r}_{\text{b},1}[k], {\bf r}_{\text{b},2}[k+1] ),
\end{align}
where $\text{GRU}_{\text{f},i}$ and $\text{GRU}_{\text{b},i}$ represent functional forms of GRU at layer $i$ in the forward and backward direction, respectively.
Here, ${\bf r}_{\text{f}, i}[k] \in \mathbb{R}^{N_{\text{r},i}}$ and ${\bf r}_{\text{b}, i}[k] \in \mathbb{R}^{N_{\text{r},i}}$ are the state vectors obtained by $\text{GRU}_{\text{f}, i}$ and $\text{GRU}_{\text{b}, i}$, respectively, at time $k$, where $i=1,2$ and $k=1,...,N$. 
For the initial conditions,
${\bf r}_{\text{f},i}[-1]= {\bf 0}$ and ${\bf r}_{\text{b},i}[N+1]= {\bf 0}$, $i=1,2$.




\textbf{Attention layer.}
We consider the state vectors at the last layer, i.e., ${\bf r}_{\text{f},2}[k]$ and ${\bf r}_{\text{b},2}[k]$, over $k=1,...,N$, as inputs to the attention layer.
Each state vector contains different feature information depending on both its direction and time-step $k$: The forward  state vector ${\bf r}_{\text{f},2}[k]$ captures the implicit correlation information of the receive signals of $y[1], ..., y[k]$, while the backward state vector ${\bf r}_{\text{b},2}[k]$ captures that of $y[k], ..., y[N]$, $k=1,...,N$. Although the state vectors at each end, i.e., ${\bf r}_{\text{f},2}[N]$ and ${\bf r}_{\text{b},2}[1]$, contain the information of all receive signals, the long-term dependency cannot be fully captured~\cite{bengio1993problem}. Therefore, we adopt the attention layer \cite{bahdanau2014neural}.
Formally,
\begin{align}
    {\bf r}_\text{f,att} &=  \sum_{k=1}^N \alpha_{\text{f},k} {\bf r}_{\text{f},2}[k],
    \quad
    {\bf r}_\text{b,att} =  \sum_{k=1}^N \alpha_{\text{b},k} {\bf r}_{\text{b},2}[k],
    \label{eq:dec:att_weight}
\end{align}
where $\alpha_{\text{f},k} \in \mathbb{R}$ and $\alpha_{\text{b},k} \in \mathbb{R}$ are the trainable \textit{attention weights} applied to the forward and backward state vectors, respectively, $k=1,...,N$.
To capture the forward and backward directional information separately, we stack the two vectors, leading to
\begin{align}
    {\bf r}_\text{att} = [ {\bf r}^T_\text{f,att}, {\bf r}^T_\text{b,att} ]^T.
    \label{eq:dec:att_vector}
\end{align}
In this separated encoder-decoder architecture,
the attention mechanism at the decoder enables the decoder to fully exploit the noisy  signal information $\{y[k]\}_{k=1}^N$.

\textbf{Non-linear layer.}
At the end of the decoder, we utilize a non-linear layer to finally obtain the estimate $\hat {\bf b}$ by using the feature vector ${\bf r}_\text{att}$ in \eqref{eq:dec:att_vector}.
The input-output relationship at the non-linear layer is given by
\begin{align}
    \hat {\bf d} =  \theta ({\bf W}_d {\bf r}_\text{att}+ {\bf v}_d)  
    \label{eq:decoder:softmax}
\end{align}
where $\theta: \mathbb{R}^{2N_{\text{r},2}} \rightarrow \mathbb{R}^M$ is an activation function, and  ${\bf W}_d \in \mathbb{R}^{ M \times 2N_{\text{r},2}}$ and ${\bf v}_d \in \mathbb{R}^{M}$ are the trainable weights and biases, respectively. In this work, we consider softmax  function for $\theta$ and set to $M=2^K$.
Then, $\hat {\bf d} \in \mathbb{R}^{2^K}$ denotes the probability distribution of $2^K$ possible outcomes of $\hat {\bf b}$.



\textbf{Model training and inference.} For training, we consider the cross entropy (CE) loss $\text{CE}({\bf d}, \hat {\bf d}) =
- \sum_{i=1}^{2^K} d_i \log {\hat d}_i$, where ${\bf d} \in \{0,1\}^{2^K}$ is the one-hot representation of~${\bf b} \in \{0,1\}^{K}$, and $d_i$ and ${\hat d}_i$ are the $i$-th entry of ${\bf d} $ and $\hat {\bf d}$, respectively.
For inference, we
 force the entry with the largest value of $\hat {\bf d}$ to 1, while setting the rest   entries to $0$, 
 and then map the obtained one-hot vector to a bit stream vector $\hat {\bf b}$.
 By treating the entire bit stream as a block through the use of one-hot vectors,
we transform our problem of minimizing BLER, i.e., $\text{Pr}[ {\bf b} \neq \hat {\bf b} ]$ in \eqref{opt:obj:state-based}, into a classification~problem.

\subsection{Modulo Approach for Longer Block Lengths}
\label{ssec:modulo}

A direct application of the proposed coding architecture for long block lengths would be infeasible, since the larger input/output sizes of the encoder and decoder result in an exponential increase in complexity. Instead, we consider a modulo approach for processing a long block length of bits by successively applying our coding architecture created for short block lengths.

We define the long block length as $L$ ($L>K$), while $K$ denotes the number of \textit{processing bits} input to our coding architecture at a time. Formally, we denote the whole bit stream as ${\bf b}_\text{long} \in \{0,1\}^L$ and the index of the bit in ${\bf b}_\text{long}$ as $\ell$, $\ell=0,...,L-1$. We consider that our feedback coding architecture has been trained with block length $K$. To process the long block length, we divide the $L$ bits into $\lfloor L/K \rfloor$ chunks, and each chunk with length $K$ is then processed with our coding architecture in a time-division manner. Formally, each chunk with length $K$ can be represented with the modulo operation as
$\big[ {\bf b}_\text{long}[\lfloor \ell/K \rfloor], 
    {\bf b}_\text{long}[\lfloor \ell/K \rfloor+1], ..., 
    {\bf b}_\text{long}[\lfloor \ell/K \rfloor+K-1] \big]$.

This modulo-based approach gives two distinct benefits. First, it reduces the complexity of the network structure by simplifying the encoding and decoding processes through successive applications of the neural network trained for a shorter block length. Second, it allows generalization across various block lengths (multiple of $K$ bits) without necessitating re-training. Our experiments in Section~\ref{ssec:sim:long} show that this approach yields substantial improvements over baseline feedback coding approaches for long block~lengths.

Note that our length-$K$ coding architecture obtains a block length gain, but the modulo approach does not provide additional block length gain beyond the length-$K$ trained neural network. The block length gain of our length-$K$ coding architecture is investigated in Appendix~\ref{sec:app:gain}.


\subsection{Computational Complexity Analysis}

We first investigate the computational complexity of our length-$K$ coding architecture for encoding/decoding $K$ bits over $N$ channel uses.
We consider the number of layers of GRUs at the encoder and decoder as $N_{e,\text{layer}}$ and $N_{d,\text{layer}}$, respectively. We assume that the same number of neurons, $N_e$ and $N_d$, is used at each layer of the encoder and decoder.
Then, the encoder and decoder of our approach will have computational complexities of $\mathcal{O}(N N_{e,\text{layer}} N_e^2   + N K N_e )$ and $\mathcal{O}(N N_{d,\text{layer}} N_d^2  + 2^K N_d )$, respectively.
We then look into the complexity for encoding/decoding $L$ ($L>K$) bits by using the length-$K$ coding architecture based on the modulo approach discussed in Section~\ref{ssec:modulo}.
The corresponding complexities for encoding and decoding are
$\mathcal{O}( \lfloor L/K \rfloor  (N N_{e,\text{layer}} N_e^2   + N K N_e) )$ and 
$\mathcal{O}( \lfloor L/K \rfloor (N N_{d,\text{layer}} N_d^2  + 2^K N_d) )$, respectively.

Under Deepcode's architecture with a single layer of RNN at the encoder and two layers of bi-GRUs at the decoder, we obtain the big-O complexities of Deepcode to be
$\mathcal{O}(K N_e^2 )$ and $\mathcal{O}(K N_d^2 )$ for encoding and decoding $K$ bits, respectively.
For encoding/decoding $L$ ($L>K$) bits, the big-O complexities are $\mathcal{O}( L N_e^2 )$ and $\mathcal{O}(L N_d^2 )$, respectively.
For encoding/decoding $K$ bits over $N$ channel uses, the linear CL coding scheme only requires $\mathcal{O}(N^2)$ and $\mathcal{O}(N)$ at the encoder and decoder, respectively.
For processing $L$ ($L>K$) bits with the CL scheme, the big-O complexities are $\mathcal{O}( \lfloor L/K \rfloor  N^2 )$ and $\mathcal{O}( \lfloor L/K \rfloor  N )$ at the encoder and decoder, respectively.
On the other hand,
for the error correction codes, decoding often imposes high computation overhead.
The complexity of turbo coding generally increases faster than linearly with the block length $L$. 
For instance, the BCJR (Bahl-Cocke-Jelinek-Raviv)  algorithm for turbo decoding~\cite{berrou1996near}
typically has complexity of $\mathcal{O}(I (LN/K)^2 )$, 
where $I$ is the number of iterations and usually larger than 10.

It is also important to note that the computations in our scheme can be parallelized, since (i) the encoding/decoding are mostly composed of matrix calculations and (ii) multiple chunks of $K$ bits are processed in a time-division manner via the modulo approach.
Overall, given the trade-off between performance and complexity, we can consider our feedback coding to improve the communication reliability at the expense of higher complexity compared with linear schemes and Deepcode. 
A detailed discussion on performance-complexity tradeoff is provided in Appendix~\ref{sec:app:layer}.

%% file: sim.tex
\section{Experimental Results}

We now present numerical experiments to validate our methodology. We measure noise power and signal to noise ratio (SNR) in decibels (dB). For brevity, details on our training procedure are relegated to Appendix~\ref{sec:app:training}. 


\subsection{Baselines}
We consider several baselines including state-of-the-art feedback schemes and well-known error correction codes.

\textbf{Repetition coding}: Each bit of ${\bf b} \in \{0,1\}^{K}$ is modulated with binary phase-shift keying (BPSK) and transmitted repetitively over $N/K$ channel uses. 

\textbf{TBCC}~\cite{ma1986tail}: We consider tail-biting convolutional coding (TBCC), adopted in LTE standards~\cite{LTEstandard} for short blocklength codes. We consider the trellis with $(7, [133,171,165])$
and BPSK modulation.

\textbf{Turbo coding}~\cite{berrou1996near}: We consider turbo codes, adopted in LTE standards for medium/long blocklength codes. 
We consider the trellis with $(4,[13, 15])$, BPSK  modulation, and $10$ decoding iterations.

\textbf{CL scheme with $2^B$-ary PAM}~\cite{chance2011concatenated}: 
The bit stream ${\bf b}$ is first divided into $K/B$ bit chunks each with length $B$. Each bit chunk is modulated with $2^B$-ary pulse amplitude modulation (PAM), and then 
transmitted over $NB/K$ channel uses with the CL scheme.

\textbf{Deepcode}~\cite{kim2018deepcode}: This is a non-linear feedback coding scheme that uses RNNs to encode/decode the bit stream {$\bf b$} based on a bit-by-bit processing. 

\textbf{TBCC with CL}: This is a concatenated coding approach with TBCC for outer coding and the CL scheme for inner coding. The bit stream ${\bf b}$ is encoded with TBCC to generate the outer code ${\bf c} \in \{0,1\}^{K/r_\text{out}}$ with outer coding rate $r_\text{out}=1/2$ and the trellis with $(7,[133,171])$. ${\bf c}$ is modulated with $2^B$-ary PAM where $B=2$,
and each symbol is 
transmitted with the CL scheme over $Nr_\text{out}B/K$
channel uses.

\textbf{Turbo with CL}:
This is a concatenated coding approach where turbo coding is used for outer coding.
For turbo coding, we consider $r_\text{out}=1/3$, the trellis with  $(4,[13, 15])$, and $10$ decoding iterations.  For the CL scheme as inner coding, 
we consider $2^B$-ary PAM where $B=2$. 



\subsection{Short Blocklength}
\label{ssec:sim:short}

\begin{figure}[t]
  \includegraphics[width=.8\linewidth]{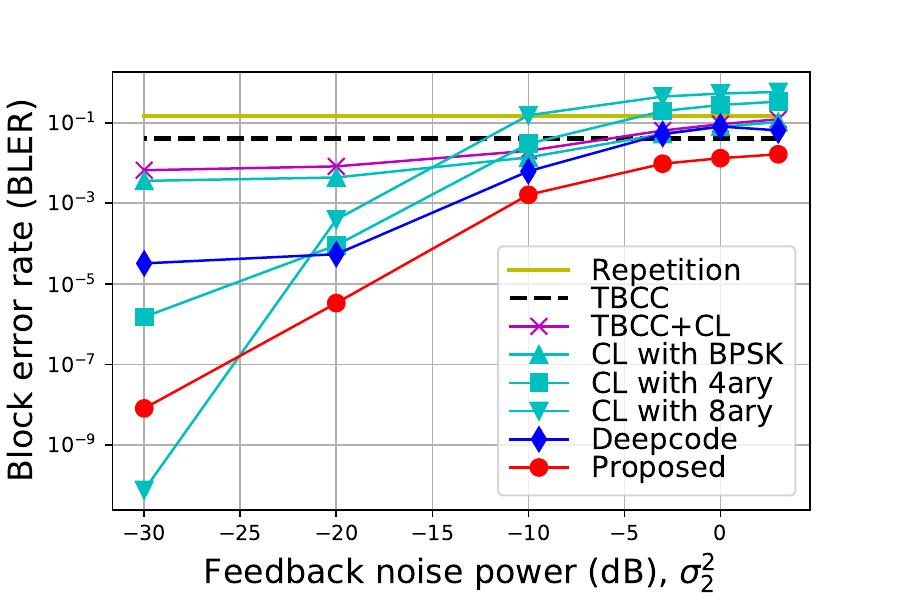}
  \centering
  \vspace{-2mm}
  \caption{BLER with short blocklength $K=6$ and rate $r=1/3$ when the forward SNR is $1$dB. Our feedback codes outperform the counterparts and demonstrate robustness to high feedback noise.}
  \label{fig:bler_K6N18}
  \vspace{-2mm}
\end{figure} 

We first investigate our approach under a short blocklength regime.
Figure \ref{fig:bler_K6N18} shows BLER curves along varying feedback noise powers, where $K=6$ and $N=18$ with rate $r=1/3$ and a forward SNR of 1dB ($\sigma_1^2 =0.794$).
Due to the lack of feedback usage, the repetition coding and TBCC provide constant BLERs along the feedback noise powers,
whereas the feedback schemes, such as Deepcode, CL, and our scheme, perform better as the feedback becomes less noisy.
Over all reasonable feedback noise regions, our scheme outperforms the alternatives, usually by more than 5dB.
Importantly, as the feedback noise increases to higher levels, i.e., $\sigma_2^2 \ge -5$dB, 
our scheme still improves the BLER performance by several dB when compared to TBCC, whereas other feedback codes are shown to be much more vulnerable to high feedback noise.
This validates the improvement in resilience to feedback noise that our scheme achieves. 
When the feedback noise becomes extremely large, our scheme behaves like an open-loop code rather than feedback codes, further demonstrated in Appendix~\ref{ssec:app:adaptability}.
Also, the results with various selections of $K$ and $r$ are consistent with the ones obtained from Figure~\ref{fig:bler_K6N18}, which are included in Appendix~\ref{ssec:app:short}.

\subsection{Extension to Medium/Long Blocklength}
\label{ssec:sim:long}

\begin{figure}[t]
  \includegraphics[width=.8\linewidth]{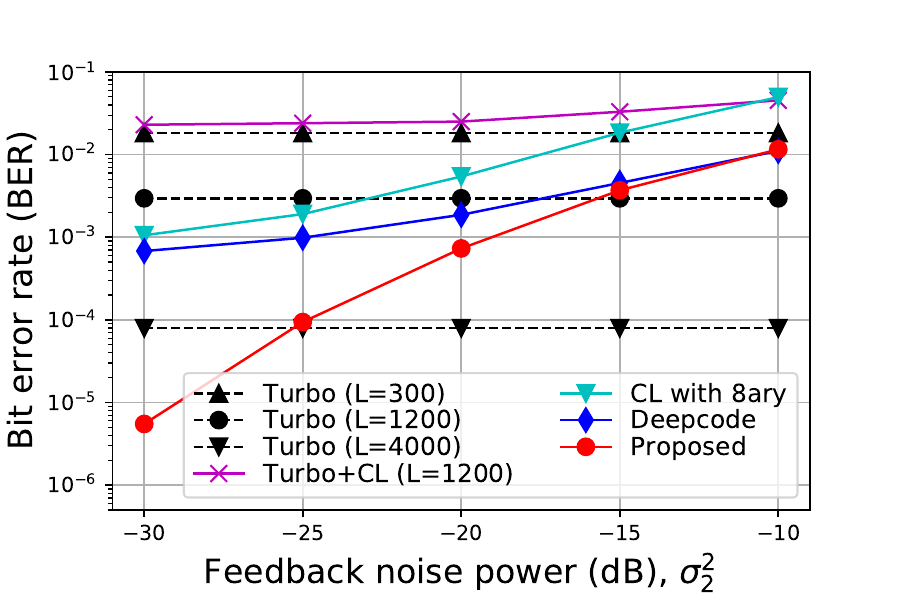}
  \centering
  \vspace{-2mm}
  \caption{BER with medium/long blocklength $L$ and rate $r=1/3$ when the forward SNR is low ($-1$dB). Our feedback codes significantly improve BER performances compared to other feedback schemes over a wide range of feedback noise powers.}
  \label{fig:ber_long2}
  \vspace{-2mm}
\end{figure}

\begin{figure}[t]
  \includegraphics[width=.8\linewidth]{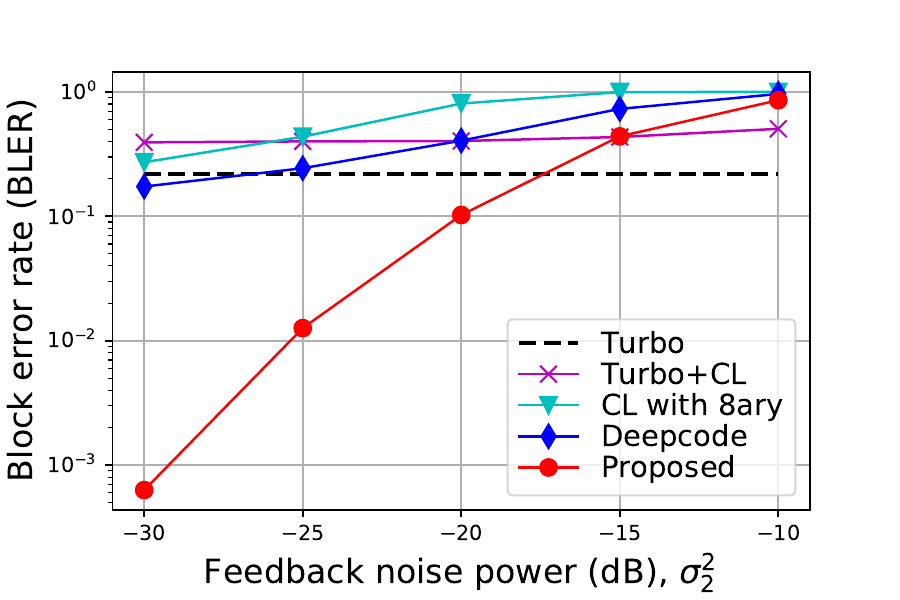}
  \centering
  \vspace{-2mm}
  \caption{BLER with medium blocklength $L=300$ and rate $r=1/3$ when the forward SNR is $-1$dB. 
  Our scheme provides a greatly enhanced region where turbo codes are outperformed.}
  \label{fig:bler_medium_L300}
  \vspace{-2mm}
\end{figure} 


We next consider a medium/long blocklength regime with rate $r=1/3$ and blocklength $L$.
%
The feedback schemes can be straightforwardly extended to the medium/long blocklength codes in a time division manner; a total of $L$ bits are divided into $L/K$ chunks of bits, and each chunk with length $K$ is processed with feedback schemes.\footnote{Although our feedback scheme can be used as inner coding for concatenated coding with a flexible inner coding rate, this strategy does not yield satisfactory results in BER and BLER for $r=1/3$. 
For Deepcode, outer coding is not allowed for rate $1/3$.}
We consider
our scheme with $K=6$ and Deepcode with $K=10$.\footnote{For Deepcode, other choices of $K$ yield similar performances.}

Figure~\ref{fig:ber_long2} shows bit error rate (BER) performances in the high forward noise scenario 
where the forward SNR is $-1$dB ($\sigma_1^2 = 1.26$).
Due to the time division processing of the feedback schemes including our scheme, Deepcode, and CL, the obtained BERs are the same for different choices of blocklength $L$, while
turbo coding benefits from the longer blocklength.
Given blocklength $L$, we identify a feedback noise power threshold, below which  feedback schemes outperform turbo coding. 
We observe that, in this high forward noise scenario, our scheme benefits significantly  from lower feedback noise power, while other feedback codes do not. 
%

Figure~\ref{fig:bler_medium_L300} shows BLERs with medium blocklength $L=300$ where the forward SNR is $-1$dB. 
Our scheme provides a better threshold while outperforming the other feedback techniques by more than two orders of magnitude in BLER.
Consistent results are observed over a variety of selections of $L$, which are presented in Appendix~\ref{ssec:app:medium}.
Nevertheless, in a very long blocklength regime, e.g., $L=4000$, turbo code may be still preferable to feedback codes unless the feedback noise is low enough, as shown in Figure~\ref{fig:ber_long2}.

\subsection{Power Control in Encoding}
\label{ssec:sim:power}

\begin{figure}[t]
  \includegraphics[width=.8\linewidth]{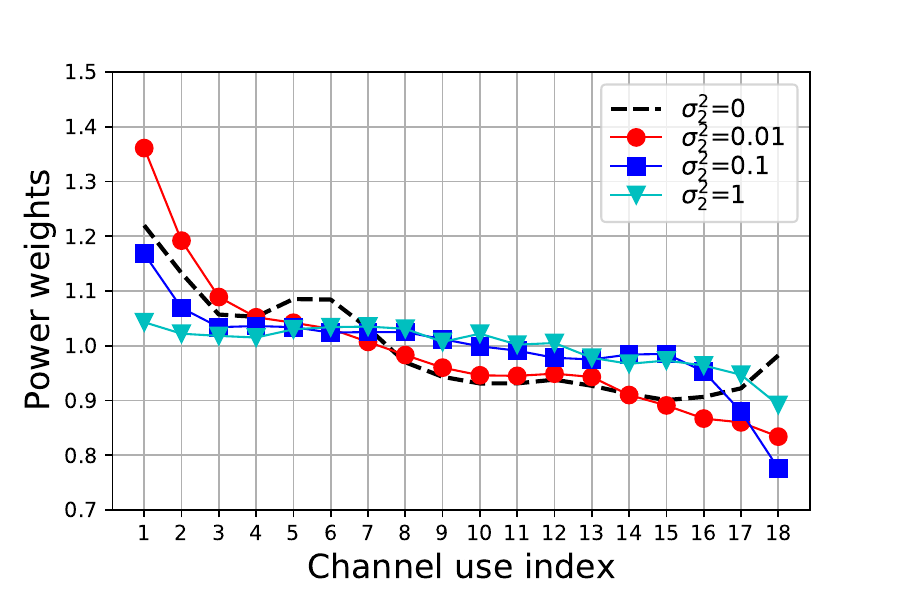}
  \centering
  \vspace{-2mm}
  \caption{Transmit power distribution over $N=18$ channel uses. The power values are larger at the beginning, and decrease along the channel uses, which aligns with the result of the CL scheme.
  }
  \label{fig:power_dist}
  \vspace{-2mm}
\end{figure} 

Figure~\ref{fig:power_dist} shows the power weights obtained by our scheme, $\{w_k\}_{k=1}^N$ in \eqref{eq:norm_power}, over $N=18$ channel uses,  under different feedback noise power scenarios.
We consider $K=6$, $r=1/3$, and a forward SNR of 1dB.
When feedback is available, the power allocation along the channel uses becomes uneven, with more powers being allocated at the start and less being distributed along later channel uses.
This result aligns with the power distribution obtained by the optimized CL linear feedback coding scheme under noisy feedback.
This shows that regardless of whether the encoding and decoding processes are linear or non-linear, more powers should be allocated at the beginning to maximize the use of the feedback information under noisy feedback.
The power distribution obtained by the CL scheme and further discussion on these points are included in Appendix~\ref{sec:app:power_dist}. 



\vspace{-2mm}
\subsection{Attention Mechanism in Decoding}
\label{ssec:sim:attention}
\vspace{-1mm}

Figure~\ref{fig:attention_weights} shows the attention weights for forward and backward directions, $\{\alpha_{\text{f},k}\}_{k=1}^N$ and $\{\alpha_{\text{b},k}\}_{k=1}^N$ in \eqref{eq:dec:att_weight},
along $N=18$ channel uses under various feedback noise power scenarios.
We consider $K=6$, $r=1/3$, and a forward SNR of 1dB.
For the low noise power scenarios, i.e., $\sigma_2^2=0.01$ or $0.1$, there is an overlap of the non-zero weight regions for forward and backward directions. This implies that all the receive signals, $y[1]$, ..., $y[N]$, are utilized at the decoder since $\alpha_{\text{f},k}>0$ means that the decoder captures the features of the receive signals, $y[1]$, ..., $y[k]$, through the forward GRUs, while $\alpha_{\text{b},k}>0$ means that the decoder captures the features of $y[k]$, ..., $y[N]$ through the backward GRUs.
On the other hand, with high feedback noise, i.e., $\sigma_2^2=1$, we have $\alpha_{\text{f},k} = \alpha_{\text{b},k} \approx 0$ at $k=7$, indicating that the receive signal $y[7]$ would not be used at the decoder. However, it is important to note that
some information of $y[7]$ is  contained in the next-time signals, $y[k]$, $k=8,9,...,N$ due to the causal encoding process that uses feedback signals as an input.
As a result, for all noise conditions, it is expected that the decoder will attempt to  fully leverage the correlation of the receive signals to
reconstruct the original bit stream.

\begin{figure}[t]
  \includegraphics[width=.8\linewidth]{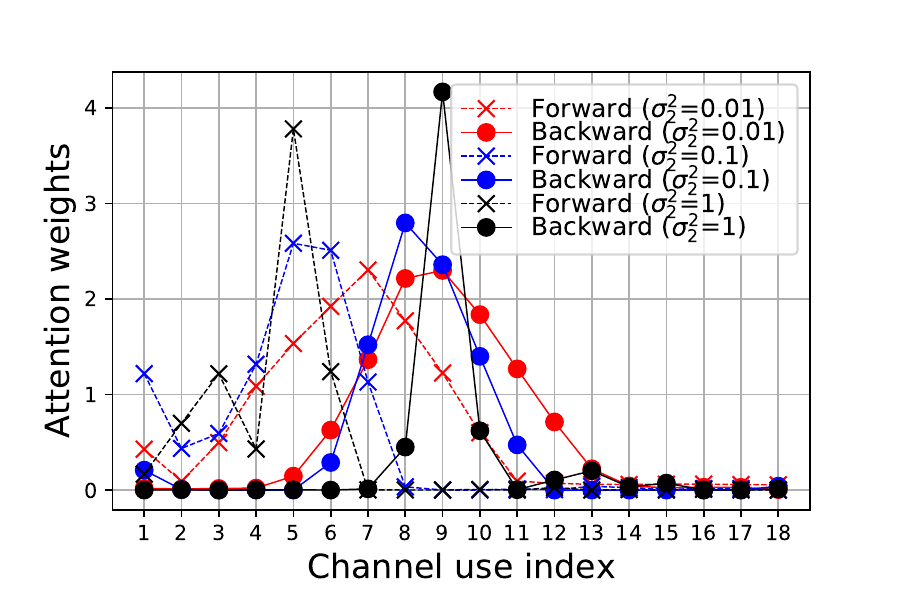}
  \centering
  \vspace{-2mm}
  \caption{Bi-directional attention weights along channel uses. Through the bi-directional GRUs and attention weights, the decoder aims to fully leverage the receive signals.}
  \label{fig:attention_weights}
  \vspace{-2mm}
\end{figure}

\begin{table}[t]
\caption{BLER performances without and with the attention layer.}
\centering
\scalebox{0.75}{
\begin{tabular}{
|c||c|c|c|c|c| }
 \hline
 & without attention layer & with attention layer \\
 \hline
 $\sigma_2^2=0.001$ &  5.25E-8 &  \textbf{8.12E-9} \\
 \hline
 $\sigma_2^2=0.005$ &  2.05E-6 &  \textbf{4.63E-7} \\
 \hline
 $\sigma_2^2=0.01$ & 4.19E-6 & \textbf{3.34E-6} \\
 \hline
 $\sigma_2^2=0.05$ &  7.15E-4 &  \textbf{4.46E-4} \\
 \hline
 $\sigma_2^2=0.1$  & 1.97E-3 & \textbf{1.63E-3} \\
 \hline
$\sigma_2^2=1$     & 1.46E-2 & \textbf{1.32E-2} \\
 \hline
\end{tabular}}
\label{table:attention}
\vspace{-2mm}
\end{table}

Table~\ref{table:attention} shows the BLER performances as the feedback noise  varies when the attention layer at the decoder is used and not used.
The attention layer provides an improvement in BLER in each noise scenario, but particularly when the feedback noise is small, e.g., a 6$\times$ reduction when $\sigma_2^2=0.001$. The attention layer captures time correlation statistics of the GRU outputs, which is shown to be more effective when the feedback channel is less noisy.

Appendix~\ref{sec:app:ablation} contains a number of ablation studies and simulations performed under various configurations, which are consistent with the findings in this section. Furthermore, Appendix~\ref{ssec:app:channel_effects} include additional experiments for realistic channel environments, including delayed feedback channels, quantized feedback channels, and Rician fading channels.

%% file: conc.tex
\vspace{-2mm}
\section{Conclusion}
\vspace{-2mm}

In this work, we presented a new class of non-linear feedback codes that significantly increase robustness against channel noise via a novel RNN autoencoder architecture.
Our learning architecture addressed the challenges of encoder-decoder separation over a noisy channel and inefficient power allocation.
To overcome encoder-decoder separation, we processed the entire bit stream as a single unit to benefit from noise averaging, and adopted a bi-directional attention-based decoding architecture to fully exploit correlations among noisy receive signals.
For power optimization, we introduced a power
control layer at the encoder, and proved that the power constraint is satisfied asymptotically.
Through numerical experiments, we demonstrated that under realistic forward/feedback noise regimes, our scheme
outperforms state-of-the-art feedback codes significantly.
We also provided information-theoretic insights on the power distribution of our non-linear feedback codes, showing that allocated power decreases over time.
%

One other important observation we made is that canonical error correction codes still outperform feedback schemes when the feedback noise becomes high in a long blocklength regime. 
In future work, we can further investigate architectural innovations to improve performance in the high noise and very long block length regime, with a more thorough comparison to canonical error correction codes.
Furthermore, it will be interesting to  look into the joint characteristics of our feedback codes in terms of error correction coding and feedback coding as mentioned in Appendix~\ref{ssec:app:adaptability} to adapt to noise conditions.

%% file: ack.tex
\vspace{-2mm}
\section*{Acknowledgements}
\vspace{-2mm}

This research was supported in part by NSF CNS2212565, NSF CNS2146171, NSF CNS2225577, NSF ITE2226447, and ONR N000142112472. 
We also thank the anonymous reviewers for their helpful feedback.

%% file: Appendix.tex
\appendix
\onecolumn

\section{Proof of Lemma 1}
\label{sec:app:proof_Lemma1}
\begin{proof}
    Define the training data tuples as $\{\mathcal{T}_j\}_{j=1}^J$ where $\mathcal{T}_j = \{\hat {\bf b}^{(j)}, \hat {\bf n}_1^{(j)}, \hat {\bf n}_2^{(j)} \}$.
    Let us denote $\tilde \eta_j[k]$ as the output at timestep $k$ generated by data $j$, $\mathcal{T}_j$, through the encoding process of \eqref{eq:s1}-\eqref{eq:enc:non-linear}, $k=1,...,N$.
    It is obvious that $\tilde \eta_j[k]$ is independent and identically distributed (i.i.d.) over $j$ assuming the data tuples $\{\mathcal{T}_j\}_{j=1}^J$ are i.i.d. from each other. We define the mean and variance of $\tilde \eta_j[k]$ as $\mu_k$ and $\nu^2_k$,  respectively.
    With the sample mean 
    $m_k(J) = \frac{1}{J}\sum_{j=1}^{J} \tilde {\eta}_j[k]$ and the sample variance 
    $\delta^2_k(J) = \frac{1}{J} \sum_{j=1}^{J} (\tilde {\eta}_j[k] - m_k(J))^2$, we define the normalization function 
    as 
    $\gamma_k^{(J)}(x) = (x-m_k(J))/\delta_k(J)$.
    
    Let us define $\tilde x[k]$ as the output at timestep $k$ generated by the data tuple for inference, $ \{{\bf b}, {\bf n}_1,  {\bf n}_2 \}$.
    Assuming the training and inference data tuples are extracted from the same distribution, the mean and variance of $\tilde x[k]$ are then $\mu_k$ and $\nu^2_k$, respectively. We then have
    $\mathbb{E}_{{\bf b}, {\bf n}_1,  {\bf n}_2} \big[ \big(  \gamma^{(J)}_k(\tilde {x}[k]) \big)^2   \big] = \frac{ \nu^2_k + (m_k(J)-\mu_k)^2}{\delta^2_k(J)}$. By strong law of large number (SLLN), $m_k(J) \rightarrow \mu_k$ and $\delta^2_k(J) \rightarrow \nu^2_k$ almost surely (a.s.) as $J \rightarrow \infty$. 
    Then, by continuous mapping theorem, $\mathbb{E}_{{\bf b}, {\bf n}_1,  {\bf n}_2} \big[ \big(  \gamma^{(J)}_k(\tilde {x}[k]) \big)^2   \big] \rightarrow 1$ a.s. as $J \rightarrow \infty$.
    Then,
$\mathbb{E}_{{\bf b}, {\bf n}_1,  {\bf n}_2} \big[ \sum_{k=1}^N \big(  {x}[k] \big)^2 \big] 
          = \sum_{k=1}^N   w_k^2 \mathbb{E}_{{\bf b}, {\bf n}_1,  {\bf n}_2} \big[ \big( \gamma_k^{(J)} (\tilde {x}[k]) \big)^2 \big]  \rightarrow N$ a.s. as $J \rightarrow \infty$.
\end{proof}

\section{Parameter Setup and Algorithm for Training}
\label{sec:app:training}

The overall algorithm for training is given in Algorithm \ref{alg:training}.
The number of training data is $J = 10^7$, the batch size is $N_{\text{batch}} = 2.5\times 10^4$, and the number of epochs is $N_{\text{epoch}} = 100$.
We use the Adam optimizer and a decaying learning rate, where the initial rate is $0.01$ and the decaying ratio is $\gamma=0.95$ applied for every epoch. 
We also use gradient clipping for training, where the gradients are clipped when the norm of gradients is larger than 1.
We adopt two layers of uni-directional GRUs at the encoder and two layers of bi-directional GRUs at the decoder, with $N_\text{neurons} = 50$ neurons at each GRU. 
We initialize each neuron in GRUs with $U(-1/\sqrt{N_\text{neurons}}, 1/\sqrt{N_\text{neurons}})$, and all the power weights and attention weights to 1.
We train our neural network model under particular forward/feedback noise powers and conduct inference in the same noise environment.

\begin{algorithm}[h!]
   \caption{Training for the proposed RNN autoencoder-based architecture}
   \label{alg:training}
\begin{algorithmic}
   \STATE {\bfseries Input:} Training data
   $\{{\bf b}^{(j)}, {\bf n}_1^{(j)},  {\bf n}_2^{(j)}  \}_{j=1}^{{J}}$, number of epochs $N_\text{epoch}$, and batch size $N_\text{batch}$.
   \STATE {\bfseries Output:} Model parameters
   \STATE Initialize the model parameters.
   \FOR{$e=1$ {\bfseries to} $N_\text{epoch}$}
   \FOR{$t=1$ {\bfseries to} $J/N_\text{batch}$}
    \STATE{Obtain $N_\text{batch}$ data tuples, $\{{\bf b}^{(\ell)}, {\bf n}_1^{(\ell)},  {\bf n}_2^{(\ell)}  \}_{\ell \in \mathcal{I}_{\text{batch}}}$, where  $\mathcal{I}_{\text{batch}}$ denotes the indices of the extracted data tuples with $\vert \mathcal{I}_{\text{batch}} \vert = N_\text{batch}$.}\\
    \STATE{Update the model parameters using the gradient decent on the cross entropy loss, defined by 
        $\sum_{\ell \in \mathcal{I}_{\text{batch}}} \text{CE}({\bf d}^{(\ell)}, \hat {\bf d}^{(\ell)})
        = - \sum_{\ell \in \mathcal{I}_{\text{batch}}} \sum_{i=1}^{2^K} d^{(\ell)}_i \log {\hat d}^{(\ell)}_i$,
    where ${\bf d}^{(\ell)} = \text{one-hot}({\bf b}^{(\ell)})$ is the target vector, and $\hat{\bf d}^{(\ell)}$ is the output obtained by the $\ell$-th data tuple passing through the overall encoder-decoder architecture in 
    \eqref{eq:s1}-\eqref{eq:decoder:softmax}.
    }
   \ENDFOR
   \ENDFOR
\end{algorithmic}
\end{algorithm}

\section{Encoder and Decoder in a Compact Form}
\label{sec:app:compact}


Figure~\ref{fig:compact} shows a compact version of the encoder and decoder, where a unrolled version along timesteps is demonstrated in Figure~\ref{fig:overall}. 
The encoder and decoder are present at the transmitter and receiver, respectively, under the framework of autoencoder architectures. The encoding output $x[k]$ travels through a noisy channel, and the noisy version of $x[k]$, i.e., $y[k]$, is provided as an input to the decoder.
Our goal with this RNN autoencoder-based structure is to jointly train the encoder and decoder.

\begin{figure}[h]
  \centering
\begin{subfigure}{.18\linewidth}
  \centering
  \includegraphics[width=\linewidth]{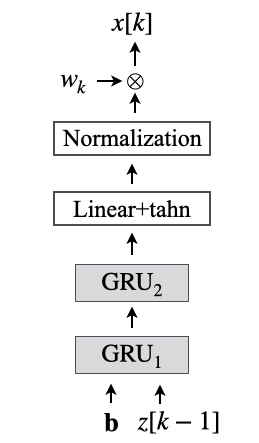}
  \caption{ Encoder
  }
  \label{fig:encoder_compact}
\end{subfigure}
\begin{subfigure}{.18\linewidth}
  \centering
  \includegraphics[width=\linewidth]{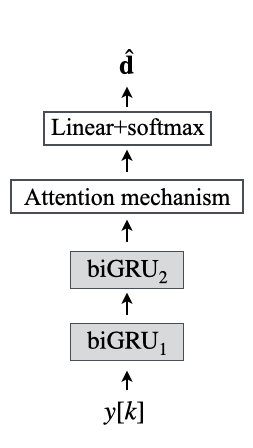}
  \caption{ Decoder
  }
  \label{fig:decoder_compact}
\end{subfigure}
  \caption{Encoder and decoder in a compact structure.}
  \label{fig:compact}
\end{figure}

\section{Evaluations of Our Scheme under Various Scenarios}

\subsection{Adaptability of our method to feedback noise}
\label{ssec:app:adaptability}

Figure~\ref{fig:bler_nofeedback} shows the BLER curves for the region of high feedback noise powers.
We consider $K=6$, $N=18$, and the forward SNR with 1dB.
Our learning architecture shown in Figure \ref{fig:overall} can be modified to generate an open-loop code that does not use any feedback signals at all, by 
providing $x[k]$ as an input to the encoder rather than $z[k]$.
This open-loop code, denoted by Proposed (No feedback) in Figure~\ref{fig:bler_nofeedback}, outperforms TBCC, showing that our learning architecture may be utilized to create error correction codes.
While keeping this as a prospective research area for future works on error correction coding, we are more interested in the behavior of our scheme in terms of feedback utilization in this work.

It is interesting that as the feedback noise power increases,
the performance of our original method (using feedback) converges to the one of our open-loop coding (without feedback utilization).
In contrast, as the feedback becomes less noisy, our scheme beats our open-loop coding scheme  by maximizing the use of the feedback information.
This demonstrates the adaptability of our approach, which may weight feedback coding or open-loop coding differently depending on the feedback noise levels.
In other words, our approach can be potentially understood as a machine learning-based joint optimization method for feedback coding and error correction coding.

\begin{figure}[h]
  \includegraphics[width=.45\linewidth]{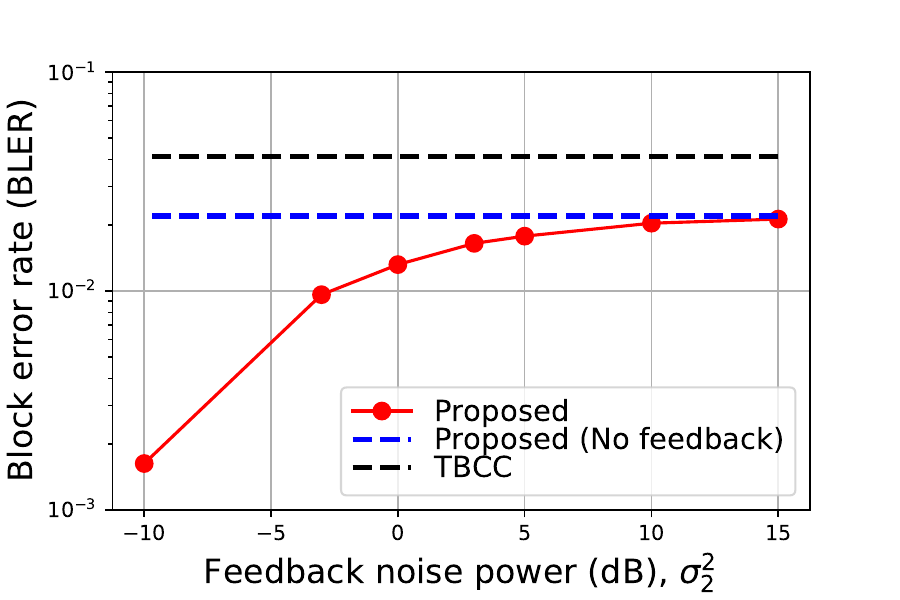}
  \centering
  \caption{BLER with $K=6$ and $N=18$ for high feedback noise, where the forward SNR is $1$dB.}
  \label{fig:bler_nofeedback}
\end{figure}

\subsection{Short blocklength scenarios}
\label{ssec:app:short}

Figure~\ref{fig:bler_K6N21} shows BLER with $K=6$, $N=21$, and the forward SNR with 1dB. 
This is the case when Deepcode benefits from the zero padding where a zero bit is added at the end of the bit stream. That is, Deepcode  utilize $18$ channel uses for transmitting  the signals corresponding to $K=6$ bits, and three channel uses for transmitting the signals corresponding to the padded zero bit.
The coding rate is then $r=2/7$.
To be fair, other schemes can use three more channel uses based on $1/3$-rate coding.
Our scheme can generate $2/7$-rate codes by simply training our RNN autoencoder-based architecture for $K=6$ and $N=21$. 
Although this scenario is favorable for Deepcode, our scheme still outperforms the baselines including Deepcode. 
Our strategy achieves significant performance gains, as we mentioned in Figure~\ref{fig:bler_K6N18}, particularly in the regions with strong feedback noise.
%
Figure~\ref{fig:bler_K10N30} demonstrates 
BLER with $K=10$ and $N=30$ ($r=1/3$). In this scenario, our scheme still outperforms the counterparts, and obtain higher performance gains over all the feedback noise regions.
Figure~\ref{fig:bler_K10N33} shows BLER with $K=10$ and $N=33$ ($r=10/33$), in which Deepcode benefits from using zero padding. Our scheme still performs better than the other baselines in realistic feedback noise power regions as well.





\begin{figure}[t]
    \centering
\begin{subfigure}{.33\linewidth}
  \includegraphics[width=\linewidth]{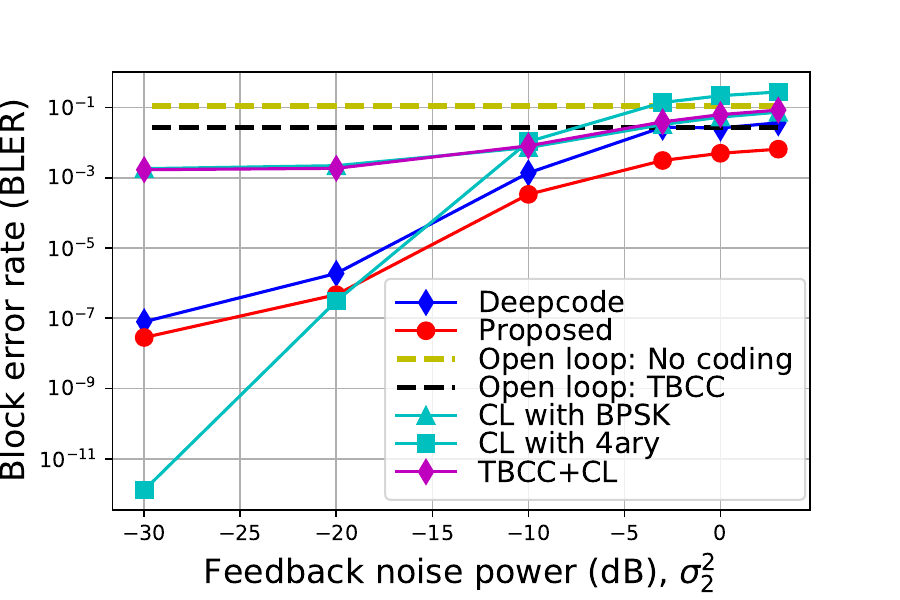}
  \centering
  \caption{$K=6$ and $N=21$ with $r=6/21$.}
  \label{fig:bler_K6N21}
\end{subfigure} 
\begin{subfigure}{.33\linewidth}
  \includegraphics[width=\linewidth]{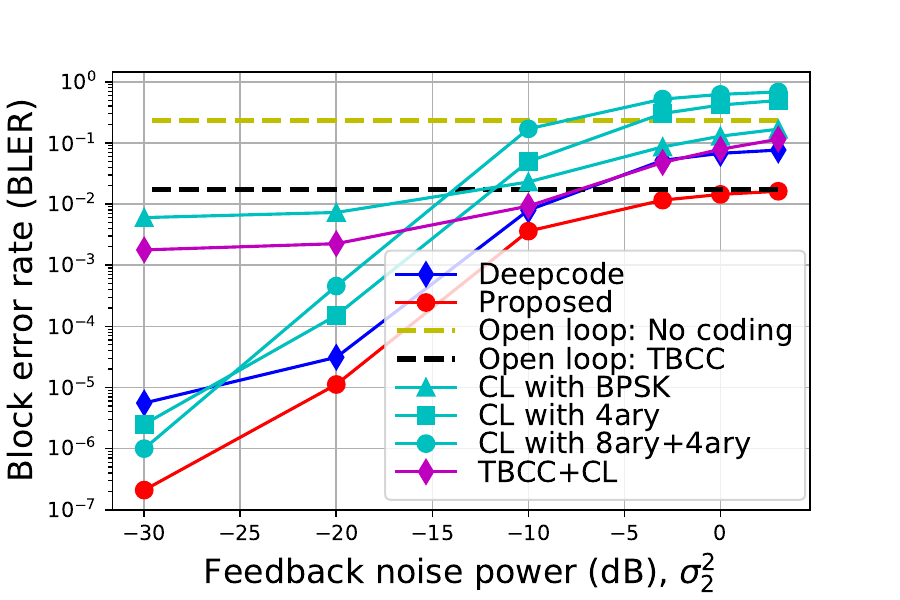}
  \centering
  \caption{$K=10$ and $N=30$ with $r=1/3$.}
  \label{fig:bler_K10N30}
\end{subfigure}
\begin{subfigure}{.33\linewidth}
  \includegraphics[width=\linewidth]{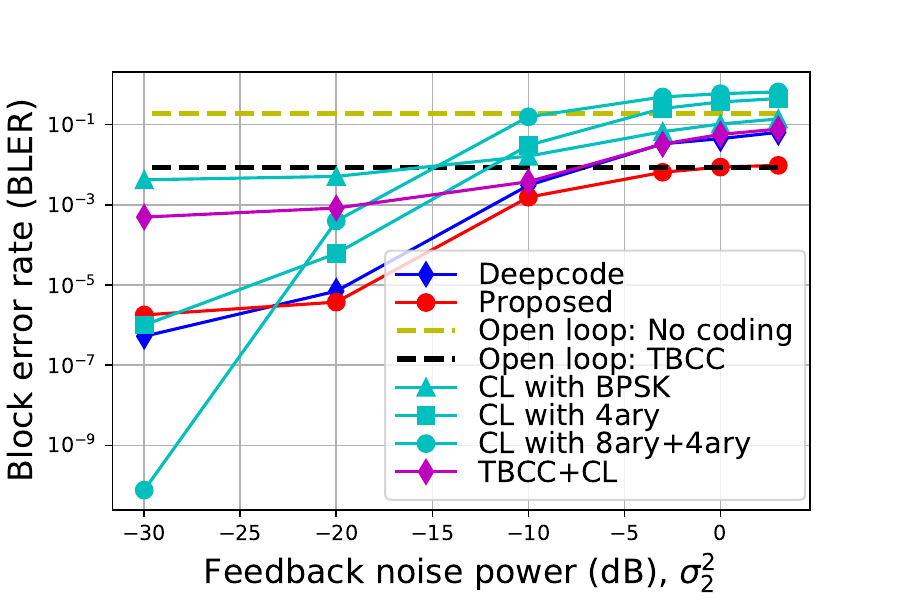}
  \centering
  \caption{$K=10$ and $N=33$ with $r=10/33$.}
  \label{fig:bler_K10N33}
\end{subfigure}
\caption{BLER curves with various selections of $K$ and $N$, where the forward SNR is $1$dB.}
\label{fig:bler_more}
\end{figure}



\subsection{Medium/long blocklength scenarios}
\label{ssec:app:medium}


Figure~\ref{fig:bler_L_60600} depicts BLER with different blocklength scenarios of $L=120, 600, 1200$,
where the forward SNR is $-1$dB and the coding rate is $r=1/3$.
Due to our scheme's tolerance to strong forward noise, it beats existing feedback methods in BLER for any blocklength scenarios by up to two orders of magnitude.
For long blocklength scenarios, such as $L=600, 1200$ in Figure~\ref{fig:bler_medium_L600} and \ref{fig:bler_medium_L1200}, any other feedback schemes other than ours cannot outperform canonical turbo coding.
However, our scheme outperforms the turbo coding depending on the feedback noise power levels. This demonstrates that, when feedback is available in  communication systems, our scheme can be used to improve the BLER performances by exploiting feedback.





\begin{figure}[h]
\centering
\begin{subfigure}{.33\linewidth}
  \includegraphics[width=\linewidth]{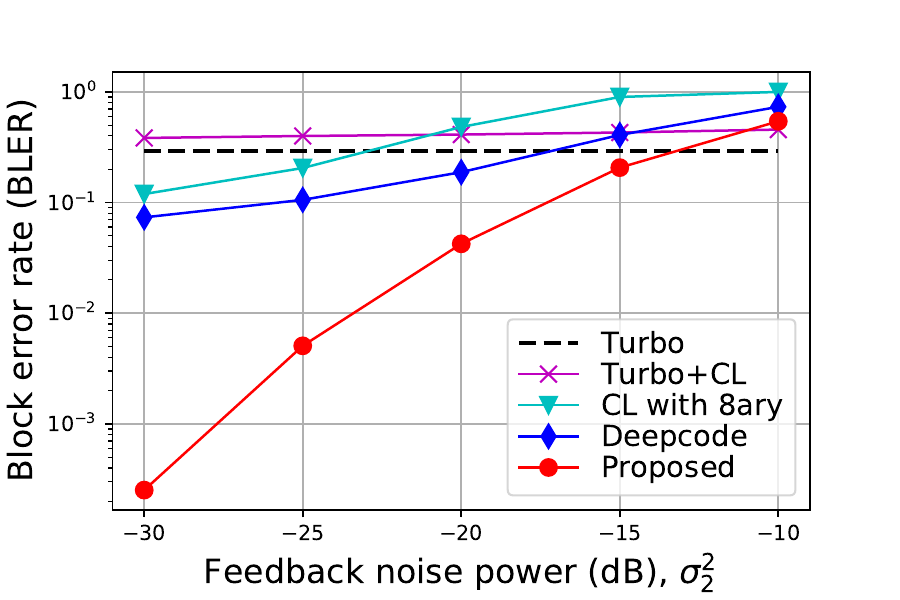}
  \centering
  \caption{$L=120$.}
  \label{fig:bler_medium_L120}
\end{subfigure} 
\begin{subfigure}{.33\linewidth}
  \includegraphics[width=\linewidth]{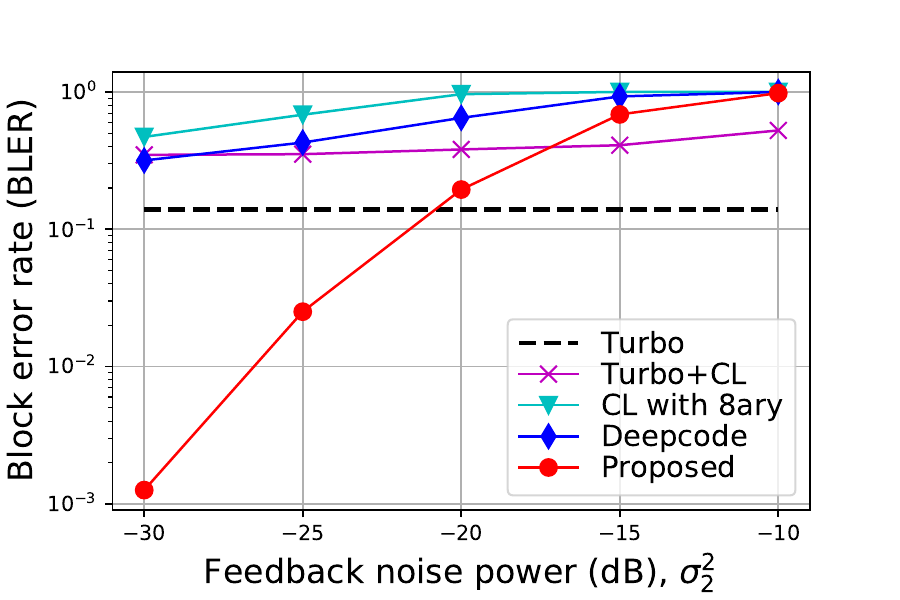}
  \centering
  \caption{$L=600$.}
  \label{fig:bler_medium_L600}
\end{subfigure} 
\begin{subfigure}{.33\linewidth}
  \includegraphics[width=\linewidth]{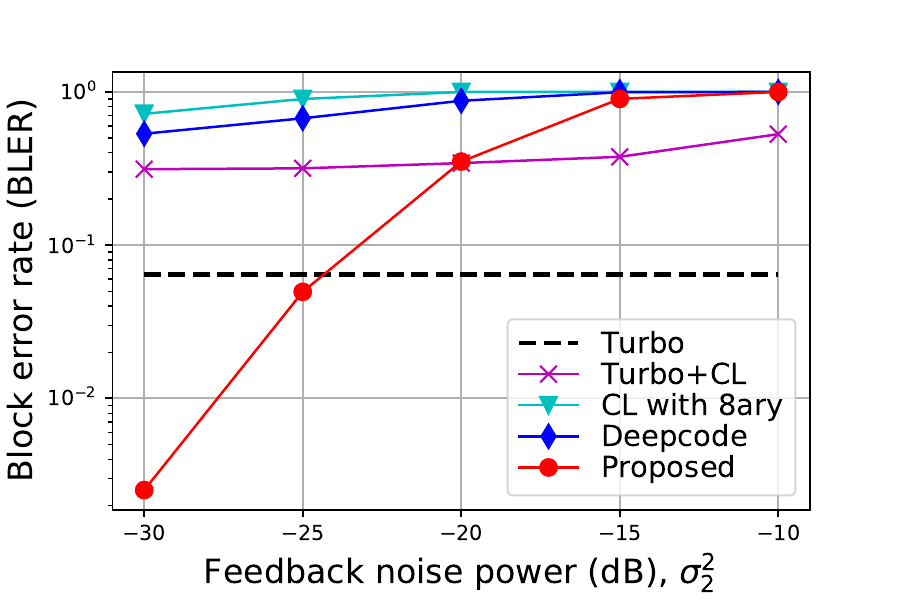}
  \centering
  \caption{$L=1200$.}
  \label{fig:bler_medium_L1200}
\end{subfigure} 
\caption{BLER of various selections of blocklength $L$ with $r=1/3$ when the forward SNR is $-1$dB.}
\label{fig:bler_L_60600}
\end{figure}

\subsection{Other realistic channel scenarios}
\label{ssec:app:channel_effects}

\subsubsection{Delayed feedback channels}

We denote the delay factor as $D$. If $D=0$, the system has no delay on feedback (our default setting), which means the transmitter receives the noisy version of the forward channel output signal in the same time step. For $D>0$, we input the previous transmit signal to the encoder in place of the feedback information until the feedback arrives. Formally, we replace

\begin{equation}
x[t] = \begin{cases}
\text{encoder}({\bf b}, 0) & \text{if} \;\; t=1
\\
\text{encoder}({\bf b}, z[t-1]) &  \text{if} \;\; 1<t \le N 
\end{cases}
\end{equation}
with
\begin{equation}
x[t] = \begin{cases}
\text{encoder}({\bf b}, 0) & \text{if} \;\; t=1
\\
\text{encoder}({\bf b}, x[t-1]) &  \text{if} \;\; 1<t \le D+1 
\\
\text{encoder}({\bf b}, z[t-D-1]) &  \text{if} \;\; D+1 <t \le N 
\end{cases}
\end{equation}
where $\text{encoder}(\cdot)$ denotes the overall process of encoding at the transmitter in Figure~\ref{fig:overall}.
Table~\ref{table:delayed_feedback} shows the BER and BLER performance with different delay factors when $K=6$, $N=18$, $\text{SNR}_1=1$dB, and $\sigma_2^2=0.01$. As expected, the error of our scheme gets larger as $D$ increases. It can tolerate a delay between $D=2$ and $D=3$ before its performance becomes comparable to Deepcode, showing it offers improved robustness to feedback delay.

\begin{table}[H]
\caption{BER and BLER under a delayed  feedback scenario with delay $D$.}
\centering
\scalebox{0.75}{
\begin{tabular}{
|c|c||c|c| }
 \hline
 & Delay & BER & BLER \\
 \hline
 Proposed & $D=0$  & 1.18E-6 & 3.34E-6\\
 \hline
 & $D=1$   & 2.41E-6 & 6.24E-6 \\
 \hline
 & $D=2$   & 2.45E-6 & 7.13E-6 \\
 \hline
 & $D=3$   & 7.67E-5 & 2.18E-4 \\
 \hline
 & $D=4$   & 8.94E-5 & 3.03E-4 \\
 \hline
 & $D=6$   & 2.24E-4 & 8.49E-4 \\
 \hline
 Deepcode & $D=0$  & 6.58E-6  & 3.74E-5\\
 \hline
\end{tabular}}
\label{table:delayed_feedback}
\end{table}

\subsubsection{Quantized feedback channels}

We consider a quantization process on feedback with $Q$ bits, where we obtain $2^Q$ points by uniformly quantizing the interval $[a_{\min}, a_{\max}]$ and map the feedback information $z[t]$ to the closest point in Euclidean distance among the $2^Q$ points. We set $a_{\min}=-5$ and $a_{\max}=5$. We test the effect of quantization on our trained coding architecture (trained with our default scenario, i.e., no quantization). 
Table~\ref{table:quantized_feedback} shows the BER and BLER performance with different numbers of quantization bits, $Q$, when $K=6$, $N=18$, $\text{SNR}_1=1$dB, and $\sigma_2^2=0.01$. We see that for $Q \ge 6$, there is not much loss in BER performance relative to the case of no quantization. Additionally, with $Q=6$ bits, our scheme outperforms Deepcode without quantization.

\begin{table}[H]
\caption{BER and BLER under a quantized  feedback scenario with $Q$-bits quantization.}
\centering
\scalebox{0.75}{
\begin{tabular}{
|c|c||c|c| }
 \hline
 &  & BER & BLER  \\
 \hline
 Proposed & No quantization  & 1.18E-6 & 3.34E-6\\
 \hline
 & $Q=8$   & 1.25E-6 & 3.44E-6 \\
 \hline
 & $Q=6$   & 3.67E-6 & 9.98E-6 \\
 \hline
 & $Q=5$   & 3.32E-5 & 9.44E-5 \\
 \hline
 & $Q=4$   & 1.65E-3 & 4.70E-3 \\
 \hline
 Deepcode & No quantization  & 6.58E-6  & 3.74E-5\\
 \hline
\end{tabular}}
\label{table:quantized_feedback}
\end{table}

\subsubsection{Rician fading channels}

Table~\ref{table:Rician} shows the BER and BLER performance obtained by our proposed scheme and Deepcode when taking into account Rician fading with different $K$-factors, where $K$-factor here refers to the ratio between the power in the direct path and the power in the scattered paths. We consider the same $K$-factors for both the forward/feedback channels. With higher $K$-factors, BER and BLER are improved since the dominant line-of-sight path behaves similar to the AWGN channel. With lower $K$-factors, on the other hand, the effective noise varies largely due to the fluctuation of the fading channel coefficients. Over all regions of $K$-factors, our coding scheme is able to improve performance over Deepcode by a wide margin.

\begin{table}[H]
\caption{BER and BLER under a Rician  fading scenario with different $K$-factors.}
\centering
\scalebox{0.75}{
\begin{tabular}{
|c|c||c|c| }
 \hline
 & $K$-factor  & BER & BLER  \\
 \hline
 Proposed & 3  & 6.85E-3 & 2.04E-2\\
 \hline
 & 10   & 1.62E-4 & 4.33E-4 \\
 \hline
 & 100   & 1.74E-6 & 4.30E-6 \\
 \hline
 & AWGN  & 1.18E-6 & 3.34E-6 \\
 \hline
 Deepcode & 3  & 9.89E-3 & 5.39E-2\\
 \hline
 & 10   & 4.99E-4 & 2.80E-3 \\
 \hline
 & 100   & 1.23E-5 & 6.89E-5 \\
 \hline
 & AWGN  & 6.58E-6 & 3.74E-5 \\
 \hline
\end{tabular}}
\label{table:Rician}
\end{table}





\section{Performance-Complexity Tradeoff}
\label{sec:app:layer}

We provide a detailed evaluation of the tradeoff between  performance and complexity. We consider the BLER performance (for measuring performance) and the FLOPS counts (for measuring complexity).
We first derive the FLOPS counts for encoding/decoding $K$ bits over $N$ channel uses.
The FLOPS counts for encoding and decoding with our coding architecture are
$N  (  6(2N_{e,\text{layer}}-1)N_e^2 + (6K + 5N_{e,\text{layer}} + 8)N_e + 3) $ and
$ 12(2N_{d,\text{layer}}-1)N N_d^2 + 2(5N_{e,\text{layer}}+8)NN_d + 4N_d2^K -2N_d + 2^K -1 $, respectively, while those with Deepcode are $2KN_e^2 + 11KN_e + 4K$ and $12K(3N_d^2 + 5N_d)$, respectively.
Table~\ref{table:ablation:layer} displays the BLER performance in the first three rows and the FLOPS counts in the last row, obtained by our coding architecture with various numbers of layers and neurons of the GRUs at the encoder and decoder, when $K=6$, $N=18$, and $r=1/3$.

\begin{table}[H]
\caption{BLER (for measuring performance) and FLOPS counts (for measuring complexity) under different number of layers and neurons.}
\centering
\scalebox{0.75}{
\begin{tabular}{
|c||c|c|c|c|c|c|c|c|c| }
 \hline
 & \multicolumn{3}{|c|}{Single layer} & \multicolumn{3}{|c|}{Two layers} & \multicolumn{2}{|c|}{Three layers} \\
 \hline
 number of neurons & 10 & 50 & 100 & 10 & 50 & 100 & 10 & 50 \\
 \hline
 $\sigma_2^2=0.01$ & 1.14E-4 & 3.16E-6 & 4.46E-6 & 5.08E-5 & 3.34E-6 & \textbf{1.32E-6} & 3.48E-5 & 4.48E-6 \\
 \hline
 $\sigma_2^2=0.1$ & 5.95E-3  & 1.85E-3 & 1.79E-3 & 5.22E-3 & \textbf{1.63E-3} & 1.94E-3 & 3.78E-3 & 2.19E-3 \\
 \hline
$\sigma_2^2=1$ & 2.78E-2     & 1.45E-2 & 1.35E-2 & 2.38E-2 &  \textbf{1.32E-2} & 1.40E-2 & 2.45E-2 & 1.33E-2\\
 \hline
 FLOPS & 48.6K	& 890.3K &	3.4M &	116.1K &	2.5M &	9.9M &	183.6K &	4.2M \\
 \hline
\end{tabular}}
\label{table:ablation:layer}
\end{table}

Overall, we find that using two layers each with 50 neurons yields strong BLER performances for this setup; increasing to 100 neurons with two layers or 50 neurons with three layers results in substantial increases in FLOPS for the same or worse BLER. While using just 10 neurons in two layers would result in a relatively low level of computational complexity, the BLER becomes substantially worse, especially for lower feedback noises. This analysis provides justification for our choice of two layers each with 50 neurons for our simulations.

On the other hand, suppose we constrain our scheme to have roughly the same complexity as Deepcode. With $K=6$, $N=18$, and $\text{SNR}_1=1$dB, if we use 1 layer with 50 neurons, our FLOPS are 890.3K, compared to 591.3K for Deepcode. For these settings, Table~\ref{table:BLER:deep} reports the BLER comparison between both methods under varying feedback noise levels. For a similar level of complexity, we still see the trend of our scheme obtaining larger improvements compared with Deepcode as the feedback noise increases.

\begin{table}[H]
\caption{BLER of our scheme and Deepcode having similar computational complexity to each other}
\centering
\scalebox{0.75}{
\begin{tabular}{
|c||c|c| }
 \hline
 & Our scheme & Deepcode \\
 \hline
 $\sigma_2^2=0.01$  & 3.16E-6 & 5.46E-5 \\
 \hline
 $\sigma_2^2=0.1$   & 1.85E-3 & 6.25E-3 \\
 \hline
$\sigma_2^2=1$      & 1.45E-2 & 8.07E-2\\
 \hline
\end{tabular}}
\label{table:BLER:deep}
\end{table}

\section{Power Distribution: Detailed Comparison with CL Scheme}
\label{sec:app:power_dist}

Figure~\ref{fig:power_dist_compare} shows power weight distribution along channel uses obtained from using our (non-linear) feedback scheme and the (linear) CL scheme, where the forward SNR is set to 1dB.
As discussed in Section~\ref{ssec:sim:power},
both linear and non-linear schemes allocate more powers at the start and less powers along timesteps under noisy feedback. 
Our scheme and the CL scheme, however, exhibit different power distribution patterns along feedback noise levels.
For the CL scheme, as the feedback link becomes less noisy, the power distribution becomes even along the channel uses. 
However, our scheme produces an even power distribution when the feedback noise becomes higher. 

The exhibition of different power distribution patterns of our scheme and the CL scheme can be understood by different encoding and decoding processes of the two schemes.
First, our scheme has a higher degree of flexibility for feedback coding due to its non-linearity, whereas the feedback codes of the CL scheme are constructed based on the linear assumption.
Second, as mentioned in Appendix~\ref{ssec:app:adaptability}, our scheme may be regarded to have some combined properties of  error correction coding and feedback coding, while the CL scheme is  specifically designed for feedback coding.
Nevertheless, it is crucial to emphasize that in order to maximize the utilization of the feedback information, both linear and non-linear schemes allocate more powers in the beginning and less powers along channel uses under a noisy feedback scenario with reasonable noise powers.



\begin{figure}[h]
  \centering
\begin{subfigure}{.45\linewidth}
    \includegraphics[width=\linewidth]{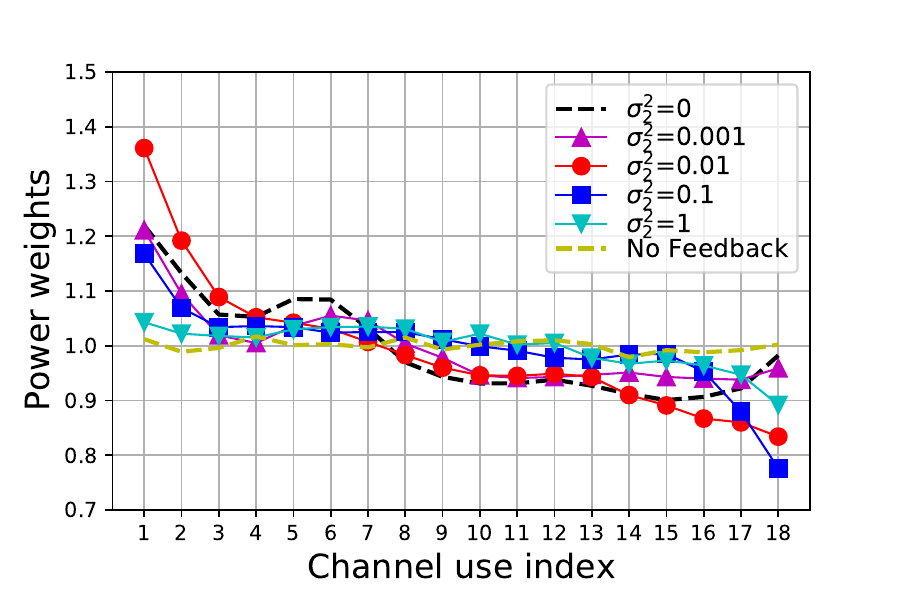}
  \centering
  \caption{Our scheme with $N=18$.}
  \label{fig:power_ours}
\end{subfigure}
\begin{subfigure}{.45\linewidth}
  \includegraphics[width=\linewidth]{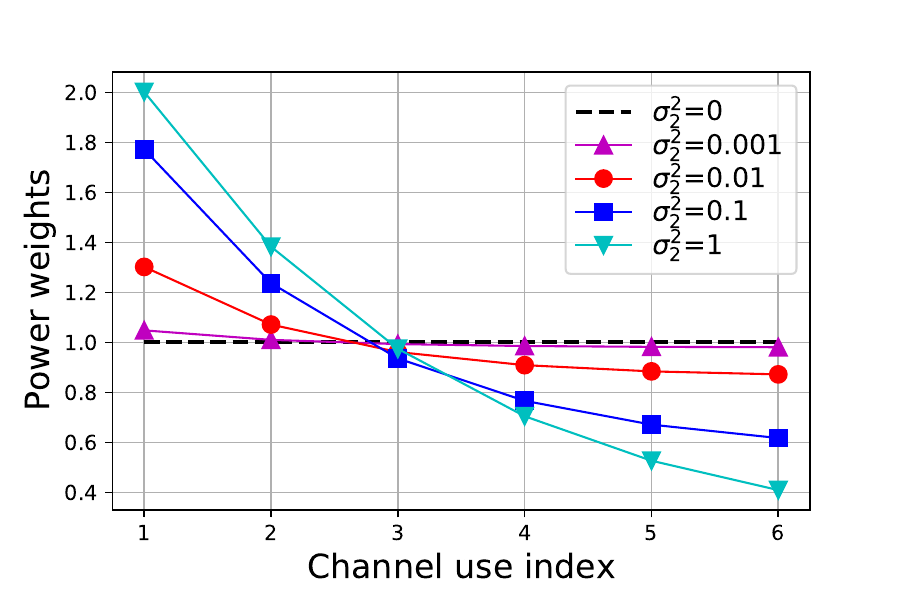}
  \centering
  \caption{CL scheme with $N=6$.}
  \label{fig:power_CL}
\end{subfigure}
\caption{Power distributions over $N$ channel uses obtained by our scheme and the CL scheme.}
\label{fig:power_dist_compare}
\end{figure}

\section{Ablation Studies and Various Configurations for Our Proposed Architecture}
\label{sec:app:ablation}

\subsection{Power control}
\label{ssec:app:power}

Table \ref{table:ablation:power} shows the BLER performances obtained from an ablation investigation of our power control scheme, where we consider $K=6$, $N=18$, and $r=1/3$.
Since the power control layer consists of the two modules of normalization and power-weight multiplication, we examine the four different cases: (i) both normalization and power-weight multiplication (the default), (ii) just normalization, (iii) just power-weight multiplication, and (iv) neither normalization nor power-weight multiplication. 
We find that the best results were obtained when both normalization and power-weight multiplication are employed.

\begin{table}[H]
\caption{Ablation study for power control.}
\centering
\scalebox{0.75}{
\begin{tabular}{
|c||c|c|c|c| }
 \hline
 & Power/norm & Norm & Power & No power/norm \\
 \hline
 $\sigma_2^2=0.01$  & \textbf{3.34E-6} & 2.87E-5 & 4.94E-2 & 6.29E-2\\
 \hline
 $\sigma_2^2=0.1$   & \textbf{1.63E-3} & 3.7E-3  & 5.59E-2 & 5.52E-2\\
 \hline
$\sigma_2^2=1$      & \textbf{1.32E-2} & 1.88E-2 & 1.26E-1 & 1.60E-1\\
 \hline
\end{tabular}}
\label{table:ablation:power}
\end{table}


\subsection{Uni-directinoal GRU versus bi-directional GRU at the decoder}

Table~\ref{table:ablation:uni_bi} shows BLER that was obtained by using our scheme either with uni-directional GRU or bi-directional GRU (the default) at the decoder, where we consider $K=6$, $N=18$, and $r=1/3$. In both cases, we consider two layers of GRUs with 50 neurons each at the encoder and decoder. 
The bi-directional case obtains higher BLERs than those from uni-directional case over the various ranges of feedback powers, despite the uni-directional case yielding comparable performances to the bi-directional case.
This is due to the possibility that a non-causal processing of the receive signals could allow the decoder to  make better use of the receive signal information.

When uni-directional GRU is used in our architecture,
we consider the uni-directional attention mechanism with $N$ forward attention weights. 
Figure~\ref{fig:dist_uni} shows the distribution of the power weights and the uni-directional attention weights.
The power distribution in the uni-directional case in Figure~\ref{fig:power_uni} is similar to the one in the bi-directional case.
In Figure~\ref{fig:att_uni}, 
the attention weights are higher at the final few channel uses, indicating that the decoder attempts to fully exploit the correlation of all the receive information by focusing on the last few outputs of the forward GRU. 
As a result,
both the uni-directional case and the bi-directional case (in Figure~\ref{fig:attention_weights}) seek to maximally utilize the correlation of the receive signals through optimizing its attention weights.


\begin{table}[H]
\caption{Comparison between uni-directional and bi-directional GRUs at the decoder.}
\centering
\scalebox{0.75}{
\begin{tabular}{
|c||c|c| }
 \hline
 & Uni-directional & Bi-directional \\
 \hline
 $\sigma_2^2=0.01$  & 5.24E-6 & \textbf{3.34E-6} \\
 \hline
 $\sigma_2^2=0.1$   & 1.73E-3 & \textbf{1.63E-3}\\
 \hline
$\sigma_2^2=1$      & 1.44E-2 & \textbf{1.31E-2}\\
 \hline
\end{tabular}}
\label{table:ablation:uni_bi}
\end{table}

\begin{figure}[h]
  \centering
\begin{subfigure}{.45\linewidth}
    \includegraphics[width=\linewidth]{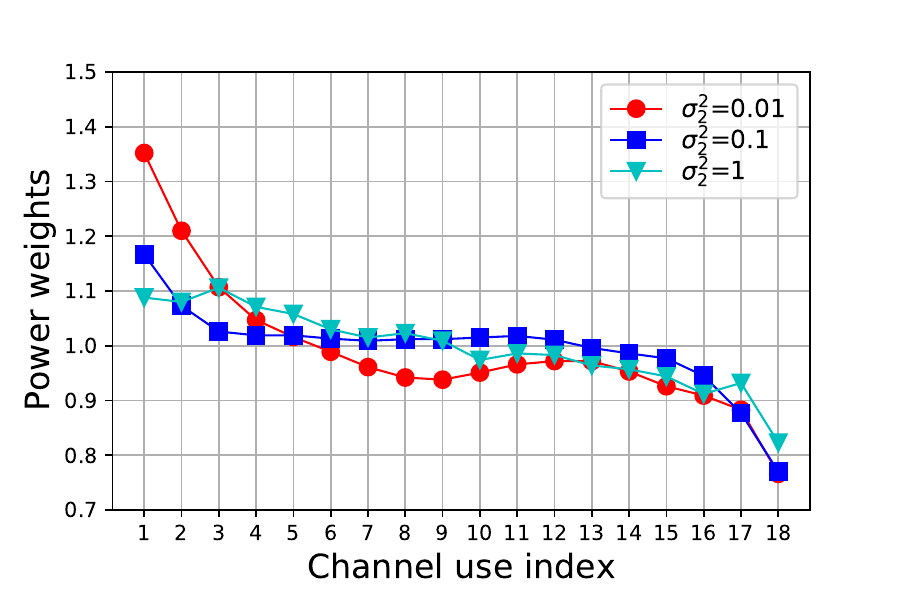}
  \centering
  \caption{Power weight distribution}
  \label{fig:power_uni}
\end{subfigure}
\begin{subfigure}{.45\linewidth}
  \includegraphics[width=\linewidth]{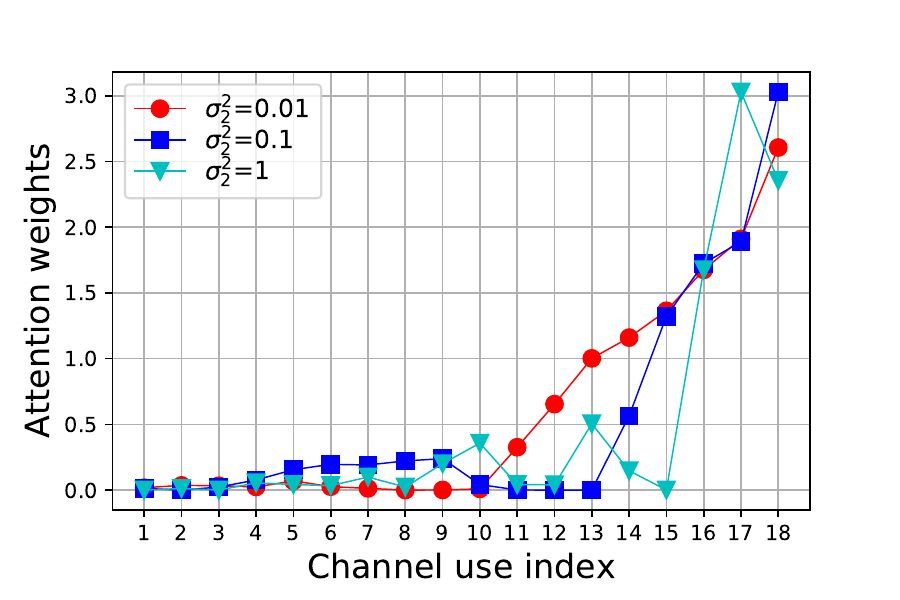}
  \centering
  \caption{Uni-directional attention weight distribution}
  \label{fig:att_uni}
\end{subfigure}
\caption{The distribution of power weights and attention weights in the case of employing uni-directional GRUs in the proposed RNN autoencoder architecture.}
\label{fig:dist_uni}
\end{figure}

\subsection{Different strategies for merging GRU outputs at the decoder}

Table~\ref{table:ablation:attention} shows the BLER performances obtained by using our scheme with various strategies for merging the bi-directional GRU outputs at the decoder, where we consider $K=6$, $N=18$, and $r=1/3$.
We explore the following five scenarios:
Case 1. We use only the hidden states at the $N$-th timestep, i.e., $\alpha_{\text{f},N}=\alpha_{\text{b},N}=1$, while the rest weights are forced to zeros. 
Case 2. We use the hidden states at the $N$-th timestep for the forward GRU and at the 1st timestep for the backward GRU, respectively. That is,
$\alpha_{\text{f},N}=\alpha_{\text{b},1}=1$, while the rest weights are zeros.
Case 3. We sum the hidden states over all the timesteps equally, i.e., $\alpha_{\text{f},k}=\alpha_{\text{b},k}=1$ for all $k=1,...,N$.
Case 4. We use the attention mechanism to put different weights along $N$ timesteps, but the same weight is applied for the forward and backward direction at each timestep. That is, $\alpha_k=\alpha_{\text{f},k}=\alpha_{\text{b},k}$ is trainable for all $k=1,...,N$.
Case 5. As our default configuration, we use a bi-directional attention mechanism where all seperate weights are applied for each timestep as well as for each of the forward and backward directions.
That is, we have $2N$ trainable weights, $\alpha_{\text{f},k}$ and $\alpha_{\text{b},k}$, for $k=1,...,N$.
We discover that in nearly all of the feedback noise power scenarios, the bi-directional attention mechanism at the decoder produced the best BLER performances.

\begin{table}[H]
\caption{Study on different strategies for merging GRU outputs at the decoder.}
\centering
\scalebox{0.75}{
\begin{tabular}{
|c||c|c|c|c|c| }
 \hline
 & Case 1 & Case  2 & Case 3 & Case 4 & Case 5 \\
 \hline
 $\sigma_2^2=0.001$ & 4.37E-7 & 5.25E-8 & 2.55E-8 & 2.31E-8 & \textbf{8.12E-9} \\
 \hline
 $\sigma_2^2=0.01$ & 1.02E-5 & 4.19E-6 & 3.78E-6 & \textbf{3.32E-6} & 3.34E-6 \\
 \hline
 $\sigma_2^2=0.1$   & 2.66E-3 & 1.97E-3 & 1.98E-3 & 1.66E-3 & \textbf{1.63E-3} \\
 \hline
$\sigma_2^2=1$     & 1.47E-2 & 1.46E-2 & 1.38E-2 & 1.36E-2 & \textbf{1.32E-2} \\
 \hline
\end{tabular}}
\label{table:ablation:attention}
\end{table}

\subsection{Sigmoid versus softmax at the last layer at the decoder}

Minimizing BLER as a performance metric for bit stream recovery has been the focus of this paper.
Other metric of interest is bit error rate (BER). Our architecture can be trained  to minimize BER rather than BLER, by simply replacing the softmax activation function with the sigmod activation function at the decoder and considering the binary cross entropy loss rather than the cross entropy loss.
Table~\ref{table:ablation:last_activation:BLER} and \ref{table:ablation:last_activation:BER} show BLER and BER, respectively, obtained by using our scheme when sigmoid or softmax is applied as an activation function at the last layer at the decoder.
We consider $K=6$, $N=18$, and $r=1/3$.
In Table~\ref{table:ablation:last_activation:BLER}, using softmax yields better BLER performances  as expected since 
it allows the neural network to be trained while minimizing the block errors.
In Table~\ref{table:ablation:last_activation:BER},
using the sigmoid is not superior to using the softmax in the noise setting $\sigma_2^2 = 0.1$ and $1$; rather, its performance is comparable to that of the softmax case.
But, at $\sigma_2^2 = 0.01$, utilizing the sigmoid yields a better BER performance. This is because the neural network with the sigmoid function is trained to minimize the bit errors. 
This demonstrates how our learning architecture can be readily modified to accommodate different metrics of interests.





\begin{table}[H]
\parbox{.5\linewidth}{
\centering
\caption{BLER with sigmoid or softmax}
\label{table:ablation:last_activation:BLER}
\scalebox{0.75}{
\begin{tabular}{
|c||c|c| }
 \hline
 & Sigmoid & Softmax  \\
 \hline
 $\sigma_2^2=0.01$  & 3.79E-6 & \textbf{3.34E-6} \\
 \hline
 $\sigma_2^2=0.1$   & 2.61E-3 & \textbf{1.63E-3} \\
 \hline
$\sigma_2^2=1$      & 3.33E-2 & \textbf{1.31E-2} \\
 \hline
\end{tabular}}
}
\hfill
\parbox{.5\linewidth}{
\centering
\caption{BER with sigmoid or softmax}
\label{table:ablation:last_activation:BER}
\scalebox{0.75}{
\begin{tabular}{
|c||c|c| }
 \hline
 & Sigmoid & Softmax  \\
 \hline
 $\sigma_2^2=0.01$  & \textbf{7.95E-7} & 1.18E-6 \\
 \hline
 $\sigma_2^2=0.1$   & 5.17E-4& \textbf{4.60E-4} \\
 \hline
$\sigma_2^2=1$      & 7.44E-3 & \textbf{6.43E-3} \\
 \hline
\end{tabular}}
}
\end{table}

\subsection{Different decoder architectures}

In Table~\ref{table:CNN}, we present the results of an additional experiment which compares the performance of our scheme, based on a bi-directional GRU decoder (our default setup), with a CNN decoder. For the CNN implementation, we employ the following standard architecture: Conv1D-Relu-Conv1D-Relu-Linear-Softmax.
The BLER performance is shown along varying $\sigma_2^2$ when $K=6$, $N=18$, and $\text{SNR}_1=1$dB. When $\sigma_2^2=0.01$, we see that the Bi-directional RNN implementation improves performance in BLER by more than 5 times. This improvement is consistent with our intuition that using GRUs at both the encoder and decoder helps to maximally exploit time correlations in the joint encoder-decoder learning architecture.

\begin{table}[H]
\caption{BLER performance of our scheme based on a bi-directional GRU decoder and a CNN decoder.}
\centering
\scalebox{0.75}{
\begin{tabular}{
|c||c|c| }
 \hline
 & bi-GRU & CNN \\
 \hline
 $\sigma_2^2=0.01$  & \textbf{3.34E-6} & 1.79E-5 \\
 \hline
 $\sigma_2^2=0.1$   & \textbf{1.63E-3} & 1.96E-3  \\
 \hline
$\sigma_2^2=1$      & \textbf{1.31E-2} & 2.34E-2 \\
 \hline
\end{tabular}}
\label{table:CNN}
\end{table}

\section{Study on Block Length Gain}
\label{sec:app:gain}


We note that our architecture promotes the block length gain within the block length $K$ processing. To investigate it, we conducted an additional simulation study by constructing our coding architecture with different values of $K$. Table~\ref{table:block_length_gain} shows the BER and BLER performance obtained by our scheme and Deepcode with different values of $K$ when $\text{SNR}_1=-1$dB and $\sigma_2^2=0.01$. Here, the BLER of $L=300$ is calculated by $\text{BLER} = 1-(1-\text{BLER}_K)^{\lfloor L/K \rfloor}$, where $\text{BLER}_K$ is the BLER of block length $K$ obtained by each coding scheme. The BER performance peaks at $K=6$ for our scheme, i.e., our proposed scheme exploits the block length gain up to $K=6$. Deepcode was shown that it does not give a block length gain for $K>50$ \cite{kim2018deepcode}.

\begin{table}[H]
\caption{BER and BLER with various  numbers of processing bits, $K$, in our coding architecture and Deepcode.}
\centering
\scalebox{0.75}{
\begin{tabular}{
|c|c||c|c| }
 \hline
 & & BER & BLER ($L=300$) \\
 \hline
 Proposed & $K=10$  & 2.28E-3 & 1.95E-1\\
 \hline
 & $K=9$   & 1.38E-3 & N/A \\
 \hline
 & $K=8$   & 8.29E-4 & N/A \\
 \hline
 & $K=7$   & 9.41E-4 & N/A \\
 \hline
 & $K=6$   & \textbf{7.32E-4} & \textbf{1.02E-1} \\
 \hline
 & $K=5$   & 2.03E-3 & 2.48E-1 \\
 \hline
 & $K=4$   & 7.44E-4 & 1.08E-1 \\
 \hline
 & $K=3$   & 6.76E-3 & 1.33E-1 \\
 \hline
 & $K=2$   & 1.64E-3 & 3.01E-1 \\
 \hline
 Deepcode & $K=50$  & 1.31E-3  & 2.98E-1\\
 \hline
 & $K=10$   & 1.34E-3 & 3.17E-1 \\
 \hline
 & $K=6$   & 1.87E-3 & 4.06E-1 \\
 \hline
\end{tabular}}
\label{table:block_length_gain}
\end{table}

For processing long block length $L$, we use the modulo approach by dividing the entire $L$ bits into $\lfloor L/K \rfloor$ chunks each with length $K$. Due to the nature of its time-division processing, the modulo approach does not provide a block length gain for processing the bits with a length of any multiples of $K$. Although our coding scheme does not provide a block length gain for processing large numbers of bits, our scheme (with a choice of $K=6$) outperforms Deepcode in BLER performance approximately by three times due to its greatly improved performance in short block length.
We note that our main focus in this paper is to provide a feedback coding framework for short/finite block lengths rather than providing a block length gain for long block length, and through simulations describe the performance gain that can be expected by taking on such an approach.
